\documentclass[reqno]{amsart}
\usepackage[foot]{amsaddr}
\usepackage[UKenglish]{babel}
\usepackage{microtype}
\usepackage[small]{complexity}
\usepackage{relsize}
\newif\ifams
\amstrue
\let\olddesc\description
\let\oldenddesc\enddescription
\providecommand{\IEEEsetlabelwidth}[1]{}
\renewenvironment{description}[1][]{\olddesc}{\oldenddesc}
\renewcommand{\paragraph}{\subsubsection*}

\makeatother
\usepackage{amssymb,mathtools,amsmath,amsthm,mathrsfs}
\usepackage{threeparttable,booktabs}
\usepackage{listings}
\usepackage[numbers,compress]{natbib}
\makeatletter
\def\NAT@spacechar{~}%
\makeatother
\usepackage[noend]{algpseudocode}

\usepackage{tikz}
\usetikzlibrary{positioning}
\usetikzlibrary{shapes.geometric}
\usetikzlibrary{automata}
\usetikzlibrary{arrows}
\tikzstyle{square}=[rounded corners=4pt,regular polygon,regular polygon
  sides=4,minimum size=.6cm,draw=gray!90,inner sep=0pt,fill=gray!20,thick]
\tikzstyle{lsquare}=[rounded corners=4pt,draw,minimum
  width=.6cm,minimum height=.45cm,draw=gray!90,inner sep=0pt,fill=gray!20,thick]
\tikzstyle{every node}=[font=\footnotesize]
\tikzstyle{every edge}=[draw,>=stealth',shorten >=1pt,thin]

\newcommand{\opfont}{\textsc}
\newcommand{\ar}{r}%
\newcommand{\X}{\textsc{Var}}
\newcommand{\N}{\?N}
\newcommand{\ru}{\?R}
\newcommand{\terms}{\opfont{Terms}_{\N}}
\newcommand{\size}[1]{\opfont{size}(#1)}
\newcommand{\vars}[1]{\opfont{var}(#1)}
\newcommand{\height}[1]{\opfont{height}(#1)}
\newcommand{\hinc}{\opfont{hinc}}
\newcommand{\sinc}{\opfont{sinc}}
\newcommand{\ntsize}[1]{\opfont{ntsize}(#1)}
\newcommand{\el}[1]{\opfont{el}(#1)}
\newcommand{\rt}[1]{\opfont{root}(#1)}
\DeclarePairedDelimiter{\tup}{(}{)}

\newcommand{\eqby}[1]{\stackrel{\raisebox{-1pt}[0pt][0pt]{\tiny #1}}{=}}
\newcommand{\eqdef}{\eqby{def}}
\newcommand{\dom}[1]{\opfont{supp}(#1)}
\newcommand{\step}[1]{\xrightarrow{#1}}
\newcommand{\dstep}[1]{\xRightarrow{#1}}
\newcommand{\lts}{\?L_\?G}
\newcommand{\rhs}{\opfont{rhs}}
\newcommand{\pairs}[1][i]{\opfont{Pairs}_{#1}}
\newcommand{\F}[1]{\ComplexityFont{F}_{\!#1}}
\newcommand{\FGH}[1]{\mathscr{F}_{\!<#1}}
\newcommand{\TOWER}{\ComplexityFont{TOWER}}

\newcommand{\ACK}{\ComplexityFont{ACKERMANN}}

\renewcommand{\eqby}[1]{\mathrel{\raisebox{-.1ex}{\ensuremath{\stackrel{\raisebox{-.25ex}{\scalebox{.5}{\upshape\textrm{#1}}}}{=}}}}}
\usepackage{url,hyperref}

\providecommand{\urlstyle}[1]{}
\providecommand{\doi}[1]{\href{http://dx.doi.org/#1}{\nolinkurl{doi:#1}}}
\usepackage{cleveref}
\newtheorem{theorem}{Theorem}
\newtheorem{lemma}[theorem]{Lemma}
\newtheorem{fact}[theorem]{Fact}
\newtheorem{corollary}[theorem]{Corollary}
\newtheorem{proposition}[theorem]{Proposition}

\theoremstyle{definition}
\newtheorem*{problem}{Problem}

\crefname{section}{Sec.}{Sections}
\Crefname{section}{Section}{Sections}
\crefname{subsection}{Sec.}{Sections}
\Crefname{subsection}{Section}{Sections}
\crefname{subsubsection}{Sec.}{Sections}
\Crefname{subsubsection}{Section}{Sections}
\crefname{theorem}{Thm.}{theorems}
\Crefname{theorem}{Theorem}{Theorems}
\crefname{lemma}{Lem.}{lemmata}
\Crefname{lemma}{Lemma}{Lemmata}
\crefname{fact}{Fact}{facts}
\Crefname{fact}{Fact}{Facts}
\crefname{corollary}{Cor.}{corollaries}
\Crefname{corollary}{Corollary}{Corollaries}
\crefname{proposition}{Prop.}{propositions}
\Crefname{proposition}{Proposition}{Propositions}
\crefname{claim}{Claim}{claims}
\Crefname{claim}{Claim}{Claims}
\crefname{definition}{Def.}{definitions}
\Crefname{definition}{Definition}{Definitions}
\crefname{example}{Ex.}{examples}
\Crefname{example}{Example}{Examples}
\crefname{remark}{Rmk.}{remarks}
\Crefname{remark}{Remark}{Remarks}
\crefname{figure}{Fig.}{figures}
\Crefname{figure}{Figure}{Figures}
\crefname{table}{Tab.}{tables}
\Crefname{table}{Table}{Tables}
\begin{document}
\title[Bisimulation of First-Order Grammars is {\textsf{ACKERMANN}}-Complete]{Bisimulation Equivalence of First-Order Grammars is {\textsf{ACKERMANN}}-Complete}
\author[P.~Jan\v{c}ar and S.~Schmitz]{Petr Jan\v{c}ar$^1$ and Sylvain Schmitz$^{2,3}$}
\address{$^1$~Dept of Computer Science, Faculty of Science\\
  Palack\'y University in Olomouc\\Czechia}
\address{$^2$~LSV, ENS Paris-Saclay \& CNRS
\\Universit\'e Paris-Saclay
\\France
}
\address{$^3$~IUF, France}

\begin{abstract}
Checking whether two pushdown automata with restricted silent actions
are weakly bisimilar was shown decidable by \citeauthor{senizergues98}
(\citeyear{senizergues98}, \citeyear{senizergues05}).  We provide the
first known complexity upper bound for this famous problem, in the
equi\-valent setting of first-order grammars.  This
\ifams{\smaller\textsf{ACKERMANN}}\else{\textbf{\textsf{ACKERMANN}}}\fi\ upper bound is optimal, and we
also show that strong bisimilarity is primitive-recursive when the
number of states of the automata is fixed.

\end{abstract}
\maketitle

\section{Introduction}
\label{sec-intro}
\emph{Bisimulation equivalence} plays a central role among the many
notions of semantic equivalence studied in verification and
concurrency theory~\citep{vanglabbeek01}.  Indeed, two bisimilar
processes always satisfy exactly the same specifications written in
modal logics~\citep{vanbenthem75} or in the modal
$\mu$-calculus~\citep{janin96}, allowing one to replace for instance a
naive implementation with a highly optimised one without breaking the
conformance.  As a toy example, the two recursive Erlang functions
below implement the same stateful message relaying service, that
either receives \lstinline[language=erlang]!{upd, M1}! and updates its
internal message from~M to~M1, or receives
\lstinline[language=erlang]!{rel,C}! and sends the message~M to the
client~C.

\ifams\hspace*{10pt}\fi\begin{minipage}[b]{\linewidth}\begin{lstlisting}[language=erlang,basicstyle=\small,literate=
               {->}{$\rightarrow{}$}{1}{,,}{\hspace{7.3pt}}{1},columns=flexible,numbers=left,firstnumber=1,numberstyle=\tiny,xleftmargin=-10pt,numbersep=2pt]
serverA(M) ->                     serverB(M) ->
  receive                           M2 = receive
    {upd, M1} -> serverA(M1);         {upd, M1} -> M1;
    {rel, C }  -> C!M,                {rel, C }  -> C!M, M;
                  serverA(M);       end,
  end.                            ,,serverB(M2).
\end{lstlisting}\end{minipage}
The two programs are weakly bisimilar if we only observe
the input (\lstinline[language=erlang]{receive}) and output
(\lstinline[language=erlang]{C!M}) actions, but the one on the left is
not tail-recursive and might perform poorly compared to the one on the
right.

In a landmark \citeyear{senizergues98} paper,
\citet{senizergues98,senizergues05} proved the decidability of
bisimulation equivalence for rooted equational graphs of finite
out-degree.  The proof extends his previous seminal
result~\citep{senizergues97,senizergues01}, which is the decidability
of language equivalence for deterministic pushdown automata (DPDA),
and entails that weak bisimilarity of pushdown processes where silent
actions are deterministic is decidable; a silent action
(also called an $\varepsilon$-step) is deterministic if it has 
no alternative when enabled.
Because the control flow of a first-order recursive program is readily
modelled by a pushdown process, one can view this result as showing
that the equivalence of recursive programs (like the two Erlang
functions above) is decidable as far as their observable behaviours
are concerned, provided silent moves are deterministic.
Regarding decidability, \citeauthor{senizergues98}' result is
optimal in the sense that bisimilarity becomes undecidable if 
we consider either
nondeterministic (popping) $\varepsilon$-steps~\citep{jancar08}, 
or second-order pushdown processes with no 
$\varepsilon$-steps~\citep{broadbent12}.  Note that the decidability
border was also refined in~\citep{yin14} by considering 
branching bisimilarity, a stronger version of weak bisimilarity.

\paragraph{Computational Complexity}

While this delineates the decidability border for equivalences of
pushdown processes, the computational complexity of the bisimilarity
problem is open.  \Citeauthor{senizergues98}' algorithm consists in
two semi-decision procedures, with no clear means of bounding its
complexity, and subsequent works like~\citep{jancar14} have
so far not proven easier to analyse.  We know however that this
complexity must be considerable, as the problem is \TOWER-hard in the
real-time case (i.e., without silent actions, hence for
\emph{strong} bisimilarity)~\citep{benedikt13} and
\ACK-hard in the general case (with deterministic silent actions)~\citep{jancarhard}---we are employing
here the `fast-growing' complexity classes defined
in~\citep{schmitz16}, where $\TOWER=\F 3$ is the
lowest non elementary class and $\ACK=\F\omega$ the
lowest non primitive-recursive one.

In fact, the precise complexity of deciding equivalences for pushdown
automata and their restrictions is often not known---as is commonplace
with infinite-state processes~\citep{srba04}.  For instance, 
language equivalence of deterministic pushdown automata is
\P-hard and was shown to be in \TOWER\ by
\citet{stirling02} (see~\citep{jancarhard} for an explicit upper
bound), and 
bisimilarity of BPAs (i.e., real-time pushdown
processes with a single state) is 
\ComplexityFont{EXPTIME}-hard~\citep{kiefer13} and in
\ComplexityFont{2EXPTIME}~\citep{burkart95} (see~\citep{jancar13} for an
explicit proof).
There are also a few known completeness results in restricted cases:
bisimilarity of
normed BPAs is \P-complete~\citep{hirshfeld96}
(see~\citep{czerwinski10} for the best known upper bound),
bisimilarity of real-time one-counter
processes (i.e., of pushdown processes with a singleton stack
alphabet) is \PSPACE-complete~\citep{bohm14}, and bisimilarity of
visibly pushdown processes is 
\ComplexityFont{EXPTIME}-complete~\citep{srba09}.

\paragraph{Contributions}

In this paper, we prove that the bisimilarity problem for pushdown
processes 
is in $\ACK$,
even the weak bisimilarity problem when
silent actions are deterministic.  
Combined with the
already mentioned lower bound 
from~\citep{jancarhard},
this shows
the problem to be \ACK-complete.  
This is the first instance of a
complexity completeness result in the line of research originating from
\citeauthor{senizergues97}'
work~\citep{senizergues97,senizergues98,senizergues01,senizergues05};
see \cref{tab-cmplx}.

\begin{table}[tbp]
  \begin{threeparttable}
  \caption{\ifams The complexity of equivalence problems over pushdown
  processes.\else The Complexity of Equivalence Problems over Pushdown Processes\fi}
  \label{tab-cmplx}
  \centering
  \begin{tabular}{lcc}  
    \toprule
    Problem              & Lower bound & Upper bound \\
    \midrule
    DPDA lang.\ equ.\!\!\!\! & \P  & \TOWER~\citep{stirling02,jancarhard}  \\
    strong bisim.\       & \!\!\TOWER~\citep{benedikt13}\!\! &\!\! \ACK~[this paper] \\
    weak bisim.\tnote{$a$} & \!\!\ACK~\citep{jancarhard}\!\!   &\!\! \ACK~[this paper] \\
    \bottomrule
  \end{tabular}
  \begin{tablenotes}
  \item[$a$] silent actions must be deterministic
  \end{tablenotes}
  \end{threeparttable}
\end{table}

Rather than working with rooted equational graphs of finite out-degree
or with pushdown processes with deterministic silent actions, our
proof is cast in the formalism of
\emph{first-order grammars} (see \cref{sec-fog}), which are term
rewriting systems with a head rewriting semantics, and are known to
generate the same class of graphs~\cite{caucal95}.

Our proof heavily relies on the main novelty from~\citep{jancar14}:
the bisimilarity of two arbitrary terms according to a first-order
grammar essentially hinges on a finite 
\emph{basis}
of pairs of \emph{non-equivalent terms}, 
which can be constructed from the grammar independently of the terms
provided as input.  The basis provides a number that allows us to
compute a bound on the `equivalence-level' of two non-equivalent
terms; this is the substance of the decision procedure
(see \cref{sec-bisim}).  Both in~\citep{jancar14} and in its reworked
version in~\citep{jancar18},
such a 
basis is obtained through a brute force argument, which yields no
complexity statement.  In \cref{sec-algo} we exhibit a concrete
algorithm computing the
basis, and we analyse its complexity
in the framework of~\citep{schmitz14,schmitz16,schmitz17}
in \cref{sec-upb}, yielding the \ACK\ upper bound.

Finally, although our results do not match the \TOWER\ lower bound
of~\citet{benedikt13} in the case of real-time pushdown processes, we
nevertheless show in \cref{sec-pda} that bisimilarity becomes
primitive-recursive in that case if additionally the number of 
control
states
of the pushdown processes is fixed.

\section{First-Order Grammars}
\label{sec-fog}
\emph{First-order grammars} are labelled term rewriting systems with a
head rewriting semantics.  They are a natural model of first-order
functional programs with a call-by-name semantics, and were shown to
generate the class of rooted equational graphs of finite out-degree by
\citet{caucal92,caucal95}, where they are called \emph{term
  context-free grammars}.  Here we shall use the terminology and
notations from~\citep{jancar18}.

\subsection{Regular Terms}

Let $\N$ be a finite ranked alphabet, i.e., where each symbol $A$
in~$\N$ comes with an arity $\ar(A)$ in~$\+N\eqdef\{0,1,2,\dots\}$,
and 
$\X\eqdef\{x_1,x_2,\dots\}$
a countable set of variables, all with
arity~zero.  We work with possibly infinite \emph{regular terms}
over~$\N$ and~$\X$, i.e., terms with finitely many distinct subterms.
Let $\terms$ denote the set of all regular terms over~$\N$ and~$\X$.
We further use $A,B,C,D$ for nonterminals, and $E,F$ for
terms, possibly primed and/or with subscripts.

\paragraph{Representations}
Such terms can be represented by finite directed graphs as shown
in \cref{fig-terms}, where each node has a label in $\N\cup\X$ and a
number of ordered outgoing arcs equal to its arity.  The unfolding of
the graph representation is the desired term, and there is a bijection
between the nodes of the \emph{least} graph representation of a
term~$E$ and the set of subterms of~$E$.
\begin{figure}[tbp]
  \centering\vspace*{-.1cm}
  \begin{tikzpicture}[auto,on grid]
    \node[square](1) {$A$};
    \node[square,below left =.75 and .7  of 1](2){$D$};
    \node[square,below      =.75         of 1](3){$x_5$};
    \node[square,below right=.75 and .7  of 1](4){$B$};
    \node[square,below left =.75 and .35 of 2](5){$x_5$};
    \node[square,below right=.75 and .35 of 2](6){$C$};
    \node[square,below left =.75 and .35 of 6](7){$x_2$};
    \node[square,below right=.75 and .35 of 6](8){$B$};
    \path[->,every node/.style={font=\tiny,inner sep=1pt,color=black!70}]
      (1) edge[swap] node {1} (2)
      (1) edge node {2}       (3)
      (1) edge node {3}       (4)
      (2) edge[swap] node {1} (5)
      (2) edge node {2}       (6)
      (6) edge[swap] node {1} (7)
      (6) edge node {2}       (8);
    \node[above left=.7 and .7 of 1](r){$\rt{E_1}$};
    \path[color=black!70]
      (r) edge                (1);
    \node[square,right=3 of 1](11) {$A$};
    \node[square,below left =.75 and .7  of 11](12){$D$};
    \node[square,below      =.75         of 11](13){$x_5$};
    \node[square,below right=.75 and .7  of 11](14){$B$};
    \node[square,below left =.75 and .35 of 12](15){$x_5$};
    \node[square,below right=.75 and .35 of 12](16){$C$};
    \node[square,below right=.75 and .35 of 16](18){$B$};
    \path[->,every node/.style={font=\tiny,inner sep=1pt,color=black!70}]
      (11) edge[swap] node {1} (12)
      (11) edge node {2}       (13)
      (11) edge node {3}       (14)
      (12) edge[swap] node {1} (15)
      (12) edge node {2}       (16)
      (16) edge node {2}       (18);
    \node[above left=.7 and .7 of 11](1r){$\rt{E_2}$};
    \path[color=black!70]
      (1r) edge                (11);
    \draw[->,>=stealth',shorten >=1pt,thin]
      (16.south west) .. controls (1.5,-4.1) and (1.5,1) .. (1.east);
    \node[font=\tiny,color=black!70,below left=.25 and .35 of 16]{1};
    \node[square,right=3 of 11](21) {$A$};
    \node[square,below left =.75 and .7  of 21](22){$D$};
    \node[square,below      =.75         of 21](23){$x_5$};
    \node[square,below right=.75 and .7  of 21](24){$B$};
    \node[square,below left =.75 and .35 of 22](25){$x_5$};
    \node[square,below right=.75 and .35 of 22](26){$C$};
    \node[square,below right=.75 and .35 of 26](28){$B$};
    \path[->,every node/.style={font=\tiny,inner sep=1pt,color=black!70}]
      (21) edge[swap] node {1} (22)
      (21) edge node {2}       (23)
      (21) edge node {3}       (24)
      (22) edge[swap] node {1} (25)
      (22) edge node {2}       (26)
      (26) edge node {2}       (28);
    \node[above left=.7 and .7 of 21](2r){$\rt{E_3}$};
    \path[color=black!70]
      (2r) edge                (21);
    \draw[->,>=stealth',shorten >=1pt,thin]
      (26.south west) .. controls (4,-4) and (4,0) .. (21.west);
    \node[font=\tiny,color=black!70,below left=.25 and .35 of 26]{1};
  \end{tikzpicture}\vspace*{-1.5cm}
  \caption{Graph representations of two finite terms $E_1$ and $E_2$,
  and of an infinite regular term~$E_3$.}
  \label{fig-terms}
\end{figure}
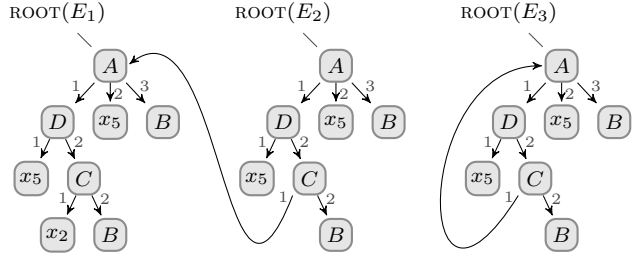

\paragraph{Size and Height}
We define the \emph{size} $\size E$ of a term~$E$ as its number of
distinct subterms.  For instance, $\size{E_1}=6$, $\size{E_2}=9$, and
$\size{E_3}=5$ in
\cref{fig-terms}.  For two terms $E$ and~$F$, we also denote by $\size{E,F}$
the number of distinct subterms of $E$ and~$F$; note that $\size{E,F}$
can be smaller than $\size E+\size F$, as they might share some
subterms.  For instance, $\size{E_1,E_2}=9$ in \cref{fig-terms}.  We
let $\ntsize E$ denote the number of distinct subterms of~$E$ with
root labels in~$\N$; e.g., $\ntsize{E_1}=4$ in \cref{fig-terms}.  
A term~$E$ is thus \emph{finite} if and only if 
its graph representation is
acyclic, in which case it has a \emph{height} $\height E$,
which is
the maximal length of a path from the root to a leaf; for instance
$\height{E_1}=3$ in \cref{fig-terms}.  Finally, we let $\vars E$
denote the set of variables occurring in~$E$, and let
$\vars{E,F}\eqdef\vars E\cup\vars F$; e.g.,
$\vars{E_1,E_2}=\{x_2,x_5\}$ in \cref{fig-terms}.

\subsection{Substitutions}

A \emph{substitution}~$\sigma$ is a map $\X\to\terms$
whose \emph{support} 
$\dom\sigma\eqdef\{x\in\X\mid \sigma(x)\neq x\}$
is finite.  This map is lifted to act 
over terms by
\begin{align*}
  x\sigma&\eqdef\sigma(x)\;,&
  A(E_1,\dots,E_{\ar(A)})\sigma&\eqdef
A(E_1\sigma,\dots,E_{\ar(A)}\sigma)
\end{align*}
for all $x$ in~$\X$, $A$ in~$\N$, and
$E_1,\dots,E_{\ar(A)}$ in~$\terms$.  For instance, in \cref{fig-terms},
$E_2=E_1\sigma$ if 
$\sigma(x_2)= E_1$ and $\sigma(x_5)= x_5$.

\subsection{Grammars}

A \emph{first-order grammar} is a tuple $\?G=\tup{\N,\Sigma,\ru}$
where $\N$ is a finite ranked alphabet of nonterminals, $\Sigma$ a
finite alphabet of actions, and $\ru$ a finite set of labelled term
rewriting rules of the form 
  $A(x_1,\dots,x_{\ar(A)})\step a E$
where
$A\in\N$, $a\in\Sigma$, and $E$ is a finite term in
$\terms$ with $\vars{E}\subseteq\{x_1,\dots,x_{\ar(A)}\}$.

\paragraph{Head Rewriting Semantics}

A first-order grammar $\?G=\tup{\N,\Sigma,\ru}$ defines an infinite
\emph{labelled transition system}
\begin{equation*}
  \lts\eqdef\tup{\terms,\Sigma,({\step a})_{a\in\Sigma}}
\end{equation*}
over $\terms$ as set of
states, $\Sigma$ as set of actions, and with a transition
relation 
${\step a}\subseteq\terms\times\terms$ for each $a\in\Sigma$,
where each rule $A(x_1,\dots,x_{\ar(A)})\step a E$ of~$\ru$ induces a
transition 
\begin{equation*}
  A(x_1,\dots,x_{\ar(A)})\sigma\step a E\sigma
\end{equation*}
for every
substitution~$\sigma$.  This means that rewriting steps can only occur
at the root of a term, rather than inside a context.  For instance,
the rules $A(x_1,x_2,x_3)\step a C(x_2,D(x_2,x_1))$ and
$A(x_1,x_2,x_3)\step b x_2$ give rise on the terms of
\cref{fig-terms} to the transitions
$E_1\step a C(x_5,D(x_5,D(x_5,C(x_2,B))))$ and $E_1\step b x_5$.
The transition relations $\step{a}$ are extended to $\step{w}$ for
words $w\in\Sigma^\ast$ in the standard way.

Note that variables $x\in\X$ are `dead', in that no
transitions can be fired from a variable.
In fact, in \cref{subsec:eqlevels} we discuss that for
technical reasons we
could formally assume that each variable~$x$ has its unique action~$a_x$
and a transition $x\step{a_x}x$.

\paragraph{Grammatical Constants}  Let us fix a first-order grammar
$\?G=\tup{\N,\Sigma,\ru}$.  We define its size as
\begin{align}\label{eq-gsize}
  |\?G|&\eqdef\sum_{A(x_1,\dots,x_{\ar(A)})\step a E\,\,\in\ru}\ar(A)+1+\size
   E\;.
    \intertext{%
Let $\rhs$ be the set of terms appearing on the right-hand sides
   of~$\ru$ (which are finite terms by definition).  We let}
 m&\eqdef\max_{A\in\N}\ar(A)\;,\label{eq-m}\\
 \hinc&\eqdef\max_{E\in\rhs}\height E-1\;,\label{eq-hinc}\\
 \sinc&\eqdef\max_{E\in\rhs}\ntsize E\label{eq-sinc}
\end{align}
bound respectively the maximal arity of its nonterminals, its maximal
height increase in one transition step, and its maximal size increase in one
transition step.

If $A(x_1,\dots,x_{\ar(A)})\step{w}x_i$ in~$\lts$ for some $i$
in~$\{1,\dots,\ar(A)\}$ and $w$ in~$\Sigma^\ast$, then we call~$w$
an \emph{$(A,i)$-sink word}. 
Observe that $w\neq\varepsilon$, hence $w=aw'$ with
$A(x_1,\dots,x_{\ar(A)})\step{a}E$ in~$\ru$ and $E\step{w'}x_i$, where
either $w'=\varepsilon$ and $E=x_i$ or $E$ `sinks' to~$x_i$ when
applying~$w'$.  Thus, for each 
$A\in\N$ and $i\in\{1,\dots,\ar(A)\}$
we can compute some shortest $(A,i)$-sink
word $w_{[A,i]}$ by dynamic programming;
in the cases where no $(A,i)$-sink
word exist, we can formally put  $w_{[A,i]}\eqdef\varepsilon$.
In turn, this
entails that the maximal length of shortest sink words satisfies
\begin{equation}\label{eq-d0}
  d_0\eqdef 1+\!\!\max_{A\in\N,1\leq i\leq\ar(A)}\!|w_{[A,i]}|\leq
  1+(2+\hinc)^{|\N|m}\;;
\end{equation}
here and in later instances, we let $\max\emptyset\eqdef 0$.

Finally, the following grammatical constant $n$ from~\citep{jancar18}
is important for us:
\begin{equation}\label{eq-n}
  n\eqdef m^{d_0}\;;
\end{equation}
note that~$n$ is at most doubly exponential in the size of~$\?G$.
This $n$ was chosen in~\citep{jancar18} so that each $E$
can be written as $E'\sigma$ where 
$\height{E'}\leq d_0$ and $\vars{E'}\subseteq\{x_1,\dots,x_n\}$, and it
is guaranteed that each path $E\step{w}F$ where $|w|\leq d_0$
can be presented as 
$E'\sigma\step{w}F'\sigma$ where $E'\step{w}F'$.
Put simply: $n$ bounds the number of depth-$d_0$ subterms for each term
$E$.

\section{Bisimulation Equivalence}
\label{sec-bisim}
Bisimulation
equivalence has been introduced independently in the
study of modal logics~\citep{vanbenthem75} and in that of concurrent
processes \citep{milner80,park81}.  We recall here the basic notions
surrounding bisimilarity before we introduce the key notion
of \emph{candidate bases} as defined in~\citep{jancar18}.

\subsection{Equivalence Levels}\label{subsec:eqlevels}

Consider a labelled transition system
\begin{equation*}
\?L=\tup{\?S,\Sigma,({\step a})_{a\in\Sigma}}
\end{equation*}
like the one defined
by a first-order grammar, with set of states~$\?S$, set of
actions~$\Sigma$, and a transition relation
${\step a}\subseteq\?S\times\?S$ for each $a$ in~$\Sigma$.  We work in
this paper with \emph{image-finite} labelled transition systems, where
$\{s'\in\?S\mid s\step a s'\}$ is finite for every~$s$ in~$\?S$ and
$a$ in~$\Sigma$.  In this setting, the coarsest (strong)
\emph{bisimulation}~$\sim$ can be defined through a chain
\begin{equation*}
{\sim_0}\supseteq{\sim_1}\supseteq\cdots\supseteq{\sim} 
\end{equation*}
of
equivalence relations over $\?S\times\?S$: let
${\sim_0}\eqdef\?S\times\?S$ and for each~$k$ in~$\+N$, let
$s\sim_{k+1} t$ if $s\sim_k t$ and
\begin{description}[\IEEEsetlabelwidth{(zag)}]
\item[(zig)] if $s\step a s'$ for some $a\in\Sigma$,
  then there exists $t'$ such that $t\step a t'$ and $s'\sim_k t'$, and
\item[(zag)] if $t\step a t'$ for some $a\in\Sigma$,
  then there exists $s'$ such that $s\step a s'$ and $s'\sim_k t'$.
\end{description}
We put
$\sim_\omega\eqdef\bigcap_{k\in\+N}{\sim_k}$; 
hence
${\sim}={\sim_\omega}$.  

For each pair $s,t$ of states in~$\?S$, we
may then define its \emph{equivalence level} $\el{s,t}$
in~$\omega+1=\+N\uplus\{\omega\}$
as
\begin{equation}\label{eq-el}
  \el{s,t}\eqdef\sup\{k\in\+N\mid s\sim_k t\}\;.
\end{equation}Here we should add that---to be
consistent with~\citep{jancar18}---we stipulate that $\el{x,E}=0$ when
$E\neq x$; in particular $\el{x_i,x_j}=0$ when $i\neq j$.  This would
automatically hold if we equipped each $x\in\X$ with a special
transition $x\step{a_x}x$ in $\lts$, as we already mentioned.  This
stipulation guarantees that $\el{E,F}\leq\el{E\sigma,F\sigma}$.

Two states $s,t$ are (strongly) \emph{bisimilar} if $s\sim t$, which
is if and only if $\el{s,t}=\omega$. 
We will later show an algorithm computing the equivalence level 
of two given terms in the labelled transition
system defined by a given first-order grammar.
The main decision problem in which we are interested 
is the following.
\begin{problem}[Bisimulation]
\hfill\\[-1.5em]\begin{description}[\IEEEsetlabelwidth{question}]
\item[input] A first-order grammar $\?G=\tup{\N,\Sigma,\ru}$ and two
  terms $E,F$ in~$\terms$.
\item[question] Is $\el{E,F}=\omega$ in the labelled transition
  system~$\lts$?
\end{description}\end{problem}  

\subsection{Bisimulation Game}\label{sub-game}
Observe that the following variant of the bisimulation problem is decidable.
\begin{problem}[Bounded Equivalence Level]
\hfill\\[-1.5em]\begin{description}[\IEEEsetlabelwidth{question}]
\item[input] A first-order grammar $\?G=\tup{\N,\Sigma,\ru}$, two
  terms $E,F$ in~$\terms$, and $e$ in~$\+N$.
\item[question] Is $\el{E,F}\leq e$ in the labelled transition
  system~$\lts$?
\end{description}\end{problem}
Indeed, as is well-known, the zig-zag condition can be recast as a
\emph{bisimulation game}  between two players called Spoiler and
Duplicator.  A position of the game is a pair $(s_1,s_2)\in\?S\times\?S$.
Spoiler wants to prove that the two states are not bisimilar, while
Duplicator wants to prove that they are bisimilar.  The game proceeds
in rounds; in each round,
\begin{itemize}
\item Spoiler chooses $i\in\{1,2\}$ and a
transition $s_i\step a s'_i$ (if no such transition exists, Spoiler
loses), then
\item Duplicator chooses a transition $s_{3-i}\step a s'_{3-i}$ with
the same label~$a$ (if no such transition exists, Duplicator loses);
\end{itemize}
the game then proceeds to the next round from position $(s'_1,s'_2)$.
Then $\el{s_1,s_2}\leq k$ if and only if Spoiler has a strategy to win
in the $(k{+}1)$th round at the latest
when starting the game from $(s_1,s_2)$.
Note that this game is determined and memoryless strategies suffice.

Thus, the bounded equivalence level problem can be solved by an
alternating Turing machine that first writes the representation of~$E$
and~$F$ on its tape, and then plays at most~$e$ rounds of the
bisimulation game, where each round requires at most a polynomial
number of computational steps in the size of the grammar (assuming a
somewhat reasonable tape encoding of the terms).

\begin{fact}  The bounded equivalence level problem is in
\label{eq-bel}
  $\ComplexityFont{ATIME}\big(\size{E,F}+\poly(|\?G|)\cdot e\big)$.
\end{fact}

\subsection{Candidate Bases}

Consider some fixed first-order grammar $\?G=\tup{\N,\Sigma,\ru}$.
Given three numbers $n$, $s$, and~$g$ in~$\+N$---which will depend on~$\?G$---, an \emph{$(n,s,g)$-candidate basis} for
non-equivalence is a set of pairs of terms
$\?B\subseteq\terms\times\terms$ associated with two sequences of
numbers $(s_i)_{0\leq i\leq n}$ and $(e_i)_{0\leq i\leq n}$ such that
\begin{enumerate}
\item  $\?B\subseteq{\nsim}$,
\item  for each $(E,F)\in\?B$ there is $i\in\{0,\dots,n\}$ such that	
	$\vars{E,F}=\{x_1,\dots,x_i\}$ and $\size{E,F}\leq s_i$,
\item $s_n\eqdef s$, and the remaining numbers are defined inductively
by
\end{enumerate}
\begin{align}
  e_i &\eqdef\max_{(E,F)\in\?B\mid\size{E,F}\leq s_i}\el{E,F}\;,\label{eq-ei}\\
  s_{i-1}&\eqdef 2s_i+g+e_i(\sinc+g)\;.\label{eq-si}
\end{align}
Note that the numbers $(s_i)_{0\leq i\leq n}$ and $(e_i)_{0\leq i\leq
n}$ are entirely determined by~$\?B$ and $n$, $s$, and $g$.
An $(n,s,g)$-candidate basis $\?B$ yields a \emph{bound}
$\?E_\?B$ defined by
\begin{equation}\label{eq-EB}
  \?E_\?B\eqdef n+1+\sum_{i=0}^n e_i\;.
\end{equation}

\paragraph{Full Bases}For $0\leq i\leq n$, let
\begin{multline}\label{eq-pairs}
  \pairs\eqdef\{(E,F)\mid\exists j\leq i\mathbin.\vars{E,F}=\{x_1,\dots,x_j\}\ifams\relax\else\\\fi\wedge
\size{E,F}\leq s_i\}\;.\end{multline} 
An $(n,s,g)$-candidate basis~$\?B$ is \emph{full below} some
equivalence level $e\in\omega+1$ if, for all $0\leq i\leq n$ and all
$(E,F)\in\pairs$ such that $\el{E,F}<e$ we have $(E,F)\in\?B$.
We
say that~$\?B$ is \emph{full} if it is full below~$\omega$.  In other
words and because $\?B\subseteq{\nsim}$, $\?B$ is full if and only
if, for all $0\leq i\leq n$, $\pairs\setminus\?B\subseteq{\sim}$.
\begin{proposition}[{\citep[Prop.~9]{jancar18}}]
  For any $n,s,g$, there is a unique full $(n,s,g)$-candidate basis,
  denoted by $\?B_{n,s,g}$.
\end{proposition}
\begin{proof}
  The full candidate basis~$\?B_{n,s,g}$ is constructed by induction
  over~$n$.  Let $s_n\eqdef s$ and consider the finite set
  $S_n\eqdef\{(E,F)\in\terms\times\terms\mid E\nsim
  F\wedge\exists j\leq n\mathbin.\vars{E,F}=\{x_1,\dots,x_j\}\wedge\size{E,F}\leq
  s_n\}$;
  $S_n$ has a maximal equivalence level
  $e_n\eqdef\max_{(E,F)\in S_n}\el{E,F}$. 
  If $n=0$, we define
  $\?B_{0,s,g}\eqdef S_0$.  Otherwise, we let $s_{n-1}\eqdef
  2s_n+g+e_n(\sinc+g)$ as in~\eqref{eq-si}; by induction hypothesis
  there is a unique full $(n-1,s_{n-1},g)$-candidate basis
  $\?B_{n-1,s_{n-1},g}$ and we set $\?B_{n,s,g}\eqdef
  S_n\cup\?B_{n-1,s_{n-1},g}$.
\end{proof}

The main result from~\citep{jancar18} can now be stated.
\begin{theorem}[{\citep[Thm.~7]{jancar18}}]\label{th-el}
  Let $\?G=\tup{\N,\Sigma,\ru}$ be a first-order grammar.  Then one
  can compute a grammatical constant~$g$ exponential in~$|\?G|$ and
  grammatical constants $n$, $s$, and $c$ doubly exponential
  in~$|\?G|$ such that, for all terms $E,F$ in $\terms$ with
  $E\nsim F$, \begin{equation*} \el{E,F}\leq
  c\cdot\big(\?E_{\?B_{n,s,g}}\cdot\size{E,F}+\size{E,F}^2\big)\;.  \end{equation*}
\end{theorem}

\Cref{th-el} therefore shows that the bisimulation problem can be
reduced to the bounded equivalence level problem, provided one can
compute the full $(n,s,g)$-candidate basis for suitable $n$, $s$,
and~$g$---see \cref{tab-cst} in the appendix for details on how the
grammatical constants $n$, $s$, $c$, and~$g$ are defined
in~\citep{jancar18}.  Our goal in \cref{sec-algo} will thus be to
exhibit a concrete algorithm computing the full candidate
basis~$\?B_{n,s,g}$, in order to derive an upper bound
on~$\?E_{\?B_{n,s,g}}$.

The proof of \citep[Thm.~7]{jancar18} relies on the following insight,
which we will also need in order to prove the correctness of our algorithm.
\begin{lemma}[{\citep[Eq.~39]{jancar18}}]\label{cl-el}
  Let $\?G=\tup{\N,\Sigma,\ru}$ be a first-order grammar, $g,n,s,c$ be
  defined as in \cref{th-el}, $E,F$ be two terms in~$\terms$ with $E\not\sim F$,
  and $\?B$ be an $(n,s,g)$-candidate basis full below~$\el{E,F}$.
  Then
  \begin{equation*}
    \el{E,F}\leq c\cdot\big(\?E_{\?B}\cdot\size{E,F}+\size{E,F}^2\big)\;.
  \end{equation*}
\end{lemma}

\section{Computing Candidate Bases}
\label{sec-algo}
\Cref{th-el} shows that, in order to solve the bisimulation problem,
it suffices to compute~$c$ and~$\?E_{\?B_{n,s,g}}$ and then solve the
bounded equivalence problem, for which \cref{eq-bel} provides a
complexity upper bound.  In this section, we show how to compute
$\?E_{\?B_{n,s,g}}$ for an input first-order grammar
$\?G=\tup{\N,\Sigma,\ru}$.  Note that this grammatical constant was
shown computable in~\cite{jancar14,jancar18} through a brute-force
argument, but here we want a concrete algorithm, whose complexity will
be analysed in \cref{sec-upb}.  We proceed in two steps, by first
considering a non effective version\ifams\relax\else\ of the
algorithm\fi\ in \cref{sub-neff}, whose correctness
is straightforward, and then the
actual algorithm in \cref{sub-eff}.

\subsection{Non Effective Version}\label{sub-neff}

Throughout this section, we consider~$n$ as a fixed parameter.  We
first assume that we have an oracle
$\textsc{EqLevel}(\?G,\?E_{\?B},c,E,F)$ at our disposal,
that returns the equivalence
level $\el{E,F}$ in~$\lts$; the parameters $\?E_{\?B},c$ will be used
in the effective version in \cref{sub-eff}.  The following procedure
then constructs  full $(n,s,g)$-candidate basis $\?B_{n,s,g}$ and its associated
bound~$\?E_{\?B_{n,s,g}}$, by progressively adding pairs from the
sets~$\pairs$ until the candidate basis is full.  In order not to
clutter the presentation too much, we assume implicitly that the
equivalence level $e$~of each pair $(E,F)$ added to~$\?B$ on
line~\ref{al-b-up} is implicitly stored, thus it does not need to be
recomputed on line~\ref{al-ei-up}.

\medskip

\begin{algorithmic}[1]%
\Procedure {CandidateBound$_n$}{$\?G$, $s$, $g$, $c$}
\State $\?B\leftarrow\emptyset$\Comment{Initialisation}\label{al-ini-start}
\For {$i\leftarrow 0,\dots,n$}
  \State $e_i\leftarrow 0$
\EndFor
\State $s_n\leftarrow s$
\For {$i\leftarrow n-1,\dots,0$}
  \State $s_i\leftarrow 2s_{i+1}+g$
\EndFor
\State $\?E_{\?B}\leftarrow n+1$\label{al-ini-end}

\For {$i\leftarrow n,\dots,0$}\label{al-pi-start}%
  \State
  $\?P_i\leftarrow\pairs\setminus\bigcup_{i<j\leq n}\?P_j$\label{al-pi-end}
\EndFor
\While{$\exists i\in\{0,1,\dots,n\}\mathbin,\exists (E,F)\in\?P_i:$
\ifams\relax\else\\
\hspace*{1.8em}\fi
$\textsc{EqLevel}(\?G,\?E_{\?B},c,E,F)<\omega$}\ifams\relax\else\Comment{Main loop}\fi
\label{al-ml-start}
  \State
  $e\leftarrow\textsc{EqLevel}(\?G,\?E_{\?B},c,E,F)$\label{al-el-inc}\ifams\Comment{Main loop}\fi
  \State $\?P_i\leftarrow \?P_i\setminus\{(E,F)\}$\label{al-pi-dec}
  \State $\?B\leftarrow \?B\cup\{(E,F)\}$\label{al-b-up}
  \If{$e > e_i$}\Comment{If so, then update}\label{al-up-start}
     \State $e_i\leftarrow e$\label{al-up-ei}
     \For{$j\leftarrow i-1,\dots,0$}\label{al-up-for}
       \State $s_j\leftarrow 2s_{j+1}+g+e_{j+1}(\sinc+g)$\label{al-up-sj}
       \State $e_j\leftarrow\max_{(E,F)\in\?B\mid\size{E,F}\leq
  s_j}\el{E,F}$\label{al-ei-up}
       \State $\?P_j\leftarrow \pairs[j]\setminus(\?B\cup\bigcup_{i<k\leq n}\?P_k)$\label{al-pi-up}
     \EndFor
     \State $\?E_{\?B}\leftarrow n+1+\sum_{0\leq j\leq n}e_j$\label{al-up-end}
  \EndIf
\EndWhile
\State\Return $\?E_{\?B}$
\EndProcedure
\end{algorithmic}

\paragraph{Invariant}

The procedure $\textsc{CandidateBound}_n$ maintains as an invariant of
its main loop on lines~\ref{al-ml-start}--\ref{al-up-end} that $\?B$
is an $(n,s,g)$-candidate basis associated with the
numbers~$(s_i)_{0\leq i\leq n}$ and~$(e_i)_{0\leq i\leq n}$, and that
$\?E_{\?B}$ is its associated bound.  This holds indeed after the
initialisation phase on lines~\ref{al-ini-start}--\ref{al-ini-end},
and is then enforced in the main loop by the update instructions on
lines~\ref{al-up-start}--\ref{al-up-end}.

\paragraph{Correctness}

Let us check that, if it terminates, this non effective version does
indeed return the bound~$\?E_{\?B_{n,s,g}}$ associated with the unique
full $(n,s,g)$-candidate basis~$\?B_{n,s,g}$.  By the previous
invariant, it suffices to show that~$\?B$ is full when the procedure
terminates.  Consider for this some index $0\leq i\leq n$ and a pair
$(E,F)\in\pairs$ with $\el{E,F}=e$ for some $e<\omega$.  By definition
of the sets $(\?P_i)_{0\leq i\leq n}$ on
lines~\ref{al-pi-start}--\ref{al-pi-end} and their updates on
lines~\ref{al-pi-dec} and~\ref{al-pi-up} in the main loop, the pair
$(E,F)$ must have been added to some~$\?P_j$ for $j\geq i$.  Then the
pair must have been selected by the condition of the main loop on
line~\ref{al-ml-start}, and added to~$\?B$.

\paragraph{Termination}

Although we are still considering a non effective version of the
algorithm, the proof that it always terminates is the same as the one
for the effective version in \cref{sub-eff}.  We exhibit a ranking
function on the main loop, thereby showing that it must stop
eventually.  More precisely, each time we enter the main loop on
line~\ref{al-ml-start}, we associate to the current state of the
procedure the ordinal rank below $\omega^{n+1}$ defined by\footnote{
Note that this is equivalent to defining the rank as the tuple
$\tup{|\?P_n|,\dots,|\?P_0|}$ in~$\+N^{n+1}$, ordered
lexicographically, but ordinal notations are more convenient for our
analysis in \cref{sec-upb}.}
\begin{equation}\label{eq-rank}
  \alpha\eqdef\omega^n\cdot|\?P_n|+\cdots+\omega^0\cdot|\?P_0|\;.
\end{equation}
This defines a descending sequence of ordinals
\begin{equation}\label{eq-ranks}
  \alpha_0>\alpha_1>\cdots
\end{equation}
of ordinals, where $\alpha_\ell$ is the rank after $\ell$~iterations of the
main loop.
Indeed, each time we enter the loop, the cardinal
$|\?P_i|$ of the set under consideration strictly decreases on
line~\ref{al-pi-dec}, and is not modified by the updates on
line~\ref{al-pi-up}, which only touch the sets~$\?P_j$ for $j<i$.
Hence $\textsc{CandidateBound}_n$ terminates.

\subsection{Effective Version}\label{sub-eff}

In order to render $\textsc{CandidateBound}_n$ effective, we provide
an implementation of $\textsc{EqLevel}$ that does not require an
oracle for the bisimulation problem, but relies instead
on \cref{cl-el} and the bounded equivalence level problem, which as we
saw in \cref{sub-game} is decidable.
\begin{algorithmic}[1]%
\Procedure{EqLevel}{$\?G$, $\?E_\?B$, $c$, $E$, $F$}
\If{$\el{E,F}\leq
c\cdot\big(\?E_\?B\cdot\size{E,F}+\size{E,F}^2\big)$}
  \State\Return $\el{E,F}$
  \Else
  \State\Return $\omega$
\EndIf
\EndProcedure
\end{algorithmic}

We establish the correctness of this effective variant in the
following theorem, which uses the same reasoning as the proof
of~\citep[Thm.~7]{jancar18}.
\begin{theorem}\label{th-algo}
  The effective version of procedure
  $\textsc{CandidateBound}_n(\?G,s,g,c)$ terminates and, provided $n$,
  $s$, $c$, and $g$ are defined as in \cref{th-el}, returns the
  bound~$\?E_{\?B_{n,s,g}}$.
\end{theorem}
\begin{proof}
  Termination is guaranteed by the ranking function defined
  by~\eqref{eq-rank}.  Regarding correctness, assume the provided $g$,
  $n$, $s$, and $c$ are defined as in \cref{th-el}, and let us define
  a (reflexive and symmetric) relation $\dot\sim_k$ on $\terms$ by
  $E\mathrel{\dot\sim_k} F$ if and only if $\el{E,F} >
  c\cdot\big(k\cdot\size{E,F}+\size{E,F}^2\big)$.  Clearly,
  ${\sim}\subseteq{\dot\sim_k}$ for all $k$~in $\+N$.  We say that an
  $(n,s,g)$-candidate basis is \emph{$k$-complete} if, for all $0\leq
  i\leq n$, $\pairs\setminus\?B\subseteq{\dot\sim_k}$.  We
  call~$\?B$ \emph{complete} if it is $\?E_{\?B}$-complete.  By the
  reasoning we used for showing the correctness of the non effective
  version, when the effective version of $\textsc{CandidateBound}_n$
  terminates, $\?B$ is complete.

  It remains to show that $\?B$ is complete if and only if it is full.
  First observe that, if $\?B$ is full, then it is complete: indeed,
  $\?B$ being full entails that,
  for all $E\nsim F$ in~$\pairs$, $(E,F)$ is
  in~$\?B\subseteq{\nsim}$, hence
  $\pairs\setminus\?B\subseteq{\sim}\subseteq{\dot\sim_{\?E_\?B}}$.
  
  Conversely, assume that~$\?B$ is complete, and let us show that it
  is full; it suffices to show that, in that case,
  ${\dot\sim_{\?E_\?B}}\subseteq{\sim}$.  By contradiction, consider a
  pair $E\nsim F$ with $E\mathrel{\dot\sim_{\?E_\?B}}F$; without loss
  of generality, $\el{E,F}$ can be assumed minimal among all such pairs.
  Then $\?B$ is full below~$\el{E,F}$: indeed, if $(E',F')\in\pairs$ and
  $\el{E',F'}<\el{E,F}$, since $\el{E,F}$ was taken minimal,
  $E'\mathrel{\dot{\nsim}_{\?E_\?B}}F'$ and therefore $(E',F')$ belongs
  to~$\?B$ since~$\?B$ is complete.  Thus \cref{cl-el} applies and
  shows that $E\mathrel{\dot\nsim_{\?E_\?B}}F$, a contradiction.
\end{proof}

\section{Complexity Upper Bounds}
\label{sec-upb}
In this section, we analyse the procedure $\textsc{CandidateBound}_n$
to derive an upper bound on the computed $\?E_{\?B}$.  In turn,
by \cref{eq-bel,th-el}, this bound will allow us to bound the
complexity of the bisimulation problem.  The idea is to analyse the
ranking function defined by \eqref{eq-rank} in order to bound how many
times the main loop of $\textsc{CandidateBound}_n$ can be executed.
We rely for this on a so-called `length function theorem'
from~\cite{schmitz14} to bound the length of descending sequences of
ordinals like~\eqref{eq-ranks}.  Finally, we classify the final upper
bound using the `fast-growing' complexity classes defined
in~\citep{schmitz16}.  A general introduction to
these techniques can be found in~\citep{schmitz17}.  
Throughout this section, we
assume that the values of $g$, $n$, $s$, and $c$ are the ones needed
for \cref{th-el} to hold.

\subsection{Controlled Descending Sequences}

Though all descending
sequences of ordinals are finite, we cannot bound their lengths in
general; e.g.,  $K+1>K>K-1>\cdots>0$ and
$\omega>K>K-1>\cdots>0$ are descending sequences of length $K+2$ for
all~$K$ in~$\+N$.  Nevertheless, the sequence~\eqref{eq-ranks}
produced by $\textsc{CandidateBound}_n$ is not arbitrary, because the
successive ranks are either determined by the input and the
initialisation phase, or the result of some computation, hence one
cannot use an arbitrary~$K$ as in these examples.

This intuition is captured by the notion of \emph{controlled
sequences}.  For an ordinal $\alpha<\omega^\omega$ (like the ranks
defined by~\eqref{eq-rank}), let us write $\alpha$ in Cantor normal
form as 
\begin{equation*}
\alpha=\omega^{n}\cdot c_n+\cdots+\omega^0\cdot c_0 
\end{equation*}
with
$c_0,\dots,c_n$ and $n$ in~$\+N$, and define its \emph{size} as
\begin{equation}
  \|\alpha\|\eqdef\max\{n,\max_{0\leq i\leq n}c_i\}\;.\label{eq-norm}
\end{equation}
Let $N_0$ be a natural number in~$\+N$ and $h{:}\,\+N\to\+N$ a
monotone inflationary function, i.e.,
$x\leq y$ implies $h(x)\leq h(y)$, and $x\leq h(x)$.
A sequence $\alpha_0,\alpha_1,\dots$
of ordinals below~$\omega^\omega$ is \emph{$(N_0,h)$-controlled} if,
for all $\ell$ in~$\+N$,
\begin{equation}\label{eq-ctrl}
  \|\alpha_\ell\|\leq h^\ell(N_0)\;,
\end{equation}
i.e., the size of the $\ell$th ordinal $\alpha_\ell$ is bounded by the
$\ell$th iterate of~$h$ applied to~$N_0$; in particular,
$\|\alpha_0\|\leq N_0$.  Because for each~$N\in\+N$, there are only
finitely many ordinals below~$\omega^\omega$ of size at most~$N$, the
length of controlled descending sequences is bounded~\citep[see,
e.g.,][]{schmitz14}.  One can actually give a precise bound on this
length in terms of \emph{subrecursive functions}, whose definition we
are about to recall.

\subsection{Subrecursive Functions}

Algorithms shown to terminate via an ordinal ranking function can have
a very high worst-case complexity.  In order to express such large
bounds, a convenient tool is found in subrecursive hierarchies, which
employ recursion over ordinal indices to define faster and faster
growing functions.  We define here two such hierarchies.

\paragraph{Fundamental Sequences}
A \emph{fundamental sequence} for a limit ordinal $\lambda$ is a
strictly ascending sequence $(\lambda(x))_{x<\omega}$ of ordinals
$\lambda(x)<\lambda$ with supremum~$\lambda$.
We use the standard assignment of fundamental
sequences to limit ordinals $\lambda\leq\omega^\omega$, defined
inductively by
\begin{align*}
  \omega^\omega(x)&\eqdef \omega^{x+1}\;,&
  (\beta+\omega^{k+1})(x)&\eqdef \beta+\omega^k\cdot(x+1)\;,
\end{align*}
where $\beta+\omega^{k+1}$ is in Cantor normal form.
This particular assignment satisfies, e.g., $0 < \lambda(x)
< \lambda(y)$ for all $x < y$. For instance, $\omega(x) = x + 1$ and
$(\omega^{3}+\omega^3+\omega)(x)=\omega^3+\omega^3+x+1$.

\paragraph{Hardy and Cicho\'n Hierarchies}
In the context of controlled sequences, the hierarchies of Hardy and
Cicho\'n turn out to be especially well-suited~\citep{cichon98}.  Let
$h{:}\,\+N\to\+N$ be a function.  For each such~$h$, the \emph{Hardy
hierarchy} $(h^\alpha)_{\alpha\leq\omega^\omega}$ and the \emph{Cicho\'n hierarchy}
$(h_\alpha)_{\alpha\leq\omega^\omega}$ relative to~$h$ are two families of functions
$h^\alpha,h_\alpha{:}\,\+N\to\+N$ defined by induction over~$\alpha$ by
\begin{align*}
  h^0(x)&\eqdef x\;,&
  h_0(x)&\eqdef 0\;,\\
  h^{\alpha+1}(x)&\eqdef h^\alpha(h(x))\;,
  &h_{\alpha+1}(x)&\eqdef1+h_\alpha(h(x))\;,\\
  h^\lambda(x)&\eqdef h^{\lambda(x)}(x)\;,&
  h_\lambda(x)&\eqdef h_{\lambda(x)}(x)\;.
\end{align*}
The Hardy functions are well-suited for expressing a large number of
iterations of the provided function~$h$.  For instance, $h^k$ for some
finite $k$ is simply the $k$th iterate of~$h$.  This intuition carries
over: $h^\alpha$ is a `transfinite' iteration of the function~$h$,
using a kind of diagonalisation in the fundamental sequences to handle
limit ordinals.  For instance, if we use the successor function $H(x)
= x+1$ as our function~$h$, we see that a first diagonalisation yields
$H^\omega(x) = H^{x+1}(x) = 2x+1$. The next diagonalisation occurs at
$H^{\omega\cdot 2}(x) = H^{\omega+x+1}(x)=H^\omega(2x + 1) = 4x +
3$. Fast-forwarding a bit, we get for instance a function of
exponential growth $H^{\omega^2}(x) = 2^{x+1} (x + 1) - 1$, and later
a non-elementary function $H^{\omega^3}$ akin to a tower of
exponentials, and a non primitive-recursive function
$H^{\omega^\omega}$ of Ackermannian growth.

In the following, we will use the following property of Hardy
functions~\citep{wainer72,cichon98}, which can be checked by induction provided
$\alpha+\beta$ is in Cantor normal form (and
justifies the use of superscripts):
\begin{gather}
  h^\alpha\circ h^\beta (x)=h^{\alpha+\beta}(x)\;,\label{eq-hardy-comp}\\
\intertext{and if $h$ is monotone inflationary, then so
is $h^\alpha$:}
\textnormal{if $x\leq y$, then } x\leq h^\alpha(x)\leq h^\alpha(y)\;.\label{eq-hardy-mono}
\end{gather}

Regarding the Cicho\'n functions, an easy induction on~$\alpha$ shows
that $H^\alpha(x) = H_\alpha(x) + x$ for the hierarchy relative to
$H(x)\eqdef x+1$.  But the main interest of
Cicho\'n functions is that they capture how many iterations are
performed by Hardy functions~\citep{cichon98}:
\begin{equation}\label{eq-hardy-cichon}
  h^{h_\alpha(x)}(x)=h^\alpha(x)\;.
\end{equation}

\paragraph{Length Function Theorem}
We can now state a `length function theorem' for controlled descending
sequences of ordinals.
\begin{theorem}[{\citep[Thm.~3.3]{schmitz14}}]\label{th-lft}
  Let $N_0\geq n+1$.  The maximal length of $(N_0,h)$-controlled
  descending sequences of ordinals in $\omega^{n+1}$ is
  $h_{\omega^{n+1}}(N_0)$.
\end{theorem}

\subsection{Controlling the Candidate Computation}

\paragraph{General Approach}

Consider an execution of $\textsc{CandidateBound}_n$ entering the main
loop at line~\ref{al-ml-start} and let us define
\begin{align}\label{eq-N}
  N&\eqdef\max\{n+1,\?E_\?B,\max_{0\leq i\leq n}s_i,\max_{0\leq i\leq
  n}|\?P_i|\}\;.
\intertext{We are going to exhibit $h{:}\,\+N\to\+N$ monotone and inflationary
such that, along any execution of $\textsc{CandidateBound}_n$, the
sequence of successive values $N_0,N_1,\dots$ defined by~\eqref{eq-N}
each time the execution enters the main loop on line~\ref{al-ml-start}
satisfies}
  N_\ell&\leq h^\ell(N_0)\label{eq-N-ctrl}
\intertext{for all~$\ell$ in~$\+N$.  By
definition of the ordinal size in~\eqref{eq-norm} of the ranks
from~\eqref{eq-rank}, $\|\alpha_\ell\|\leq N_\ell$.  Hence, this will
show that the corresponding sequence of ranks
$\alpha_0>\alpha_1>\cdots$ is $(N_0,h)$-controlled.
Therefore, \cref{th-lft} can be applied since furthermore $N_0\geq
n+1$, showing that the number of loop iterations is bounded by}
  L&\eqdef h_{\omega^{n+1}}(N_0)\;.
\intertext{By \eqref{eq-hardy-cichon}, this will entail an upper bound
on the returned~$\?E_\?B$ when $\textsc{CandidateBound}_n$ terminates:}
  \?E_\?B&\leq N_L\leq h^L(N_0)=h^{\omega^{n+1}}(N_0)\;.\label{eq-EB-h}
\end{align}

\paragraph{Controlling one Loop Execution}

As a preliminary, let us observe that, for all $0\leq i\leq n$, the
number of elements of $\pairs$ (defined in~\eqref{eq-pairs}) can be
bounded by
\begin{equation}\label{eq-pairs-size}
  |\pairs|\leq \big((|\N|+i)\cdot s_i^m\big)^{s_i}\cdot s_i^2\leq
   2^{3s_i|\?G|\log n\log s_i}\;.
\end{equation}
Indeed, the graph representation of some pair $(E,F)$ in $\pairs$ has
at most~$s_i$ vertices,
each labelled by a nonterminal symbol
from~$\N$ or a variable from~$\{x_1,\dots,x_i\}$ and with at most~$m$
outgoing edges; finally the two roots must be distinguished.

Let us turn our attention to the contents of the main loop.
\begin{lemma}\label{lem-ctrl}
For all $\ell$ in~$\+N$ we have $N_{\ell+1}\leq G_{\?G}(N_\ell)$ where
  $$G_{\?G}(x)\eqdef 2^{2^{2n+6}c^2g^2|\?G|^3x^4}\;.$$
\end{lemma}
\begin{proof}
  Assume we enter the main loop for the $\ell$th time with $N_\ell$ as
  defined in~\eqref{eq-N}.  On line~\ref{al-el-inc}, a new equivalence
  level~$e$ is introduced, with $e\leq 2cN_\ell^2$ since $\?E_\?B\leq
  N_\ell$ and $\size{E,F}\leq N_\ell$, thus in case of an update on
  line~\ref{al-up-ei}, we have $e_i\leq 2cN_\ell^2$.
  Consider now the \textbf{for} loop on
  lines~\ref{al-up-for}--\ref{al-pi-up}.  Regarding
  line~\ref{al-ei-up}, observe that
  $\max_{(E,F)\in\?B}\el{E,F}\leq\max\{e,\?E_{\?B}\}\leq 2cN_\ell^2$,
  thus
  \begin{align}
    e_j&\leq 2cN_\ell^2\label{eq-ej}
  \intertext{for all $j$ in~$\{i,\dots,0\}$ and
  $s_i\leq N_\ell$ by assumption.  Thus, regarding
  line~\ref{al-up-sj}, for all $j$ in~$\{i-1,\dots,0\}$,}
    s_j&\leq 2^{i-j}N_\ell+(2^{i-j}-1)(g+2cN_\ell^2(\sinc+g))\notag\\
       &\leq 2^{n+2}cg|\?G|N_\ell^2\;.\label{eq-sj}
  \intertext{%
  Regarding line~\ref{al-pi-up}, by \eqref{eq-pairs-size}, \eqref{eq-sj}
  entails that for all $j$ in~$\{i-1,\dots,0\}$,}
    |\?P_j|&\leq 2^{2^{2n+6}c^2g^2|\?G|^3N_\ell^4}\;.
  \end{align}
  Finally, regarding line~\ref{al-up-end}, by~\eqref{eq-ej}, $\?E_\?B\leq
  2(n+1)cN_\ell^2$.  
\end{proof}

\paragraph{Final Bound}

Let us finally express \eqref{eq-EB-h} in terms of~$n$ and~$|\?G|$.
First observe that, at the end of the initialisation phase of
lines~\ref{al-ini-start}--\ref{al-ini-end}, $e_i=0$, $s_i\leq
2^{n+1}g$, $|\?P_i|\leq 2^{2^{2n+5}s^2g^2\log|\?G|}$, and $\?E_\?B=n+1$, thus
\begin{equation}\label{eq-N0}
  N_0\leq 2^{2^{2n+5}s^2g^2\log|\?G|}\;.
\end{equation}

Then, because the bounds in \cref{lem-ctrl,eq-N0} are in terms of
$|\?G|$ (recall that the grammatical constant $g$ is exponential and
$n$, $s$, and $c$ are doubly exponential in terms of~$|\?G|$), there
exists a constant~$d$ independent from~$\?G$ such that
$|\?G|\leq N_0\leq H^{\omega^2\cdot d}(|\?G|)$ and
$G_{\?G}(x)\leq H^{\omega^2\cdot d}(\max\{x,|\?G|\})$ for all~$\?G$
and~$x$, where according to~\eqref{eq-hardy-comp}
$H^{\omega^2\cdot d}$ is the $d$th iterate of
$H^{\omega^2}(x)=2^{x+1}(x+1)-1$.  Then by~\eqref{eq-hardy-mono},
$h(x)\eqdef H^{\omega^2\cdot d}(x)$ is a suitable control function
that satisfies~\eqref{eq-N-ctrl} and therefore~\eqref{eq-EB-h}.

Finally, because $n\leq N_0\leq h(|\?G|)$ and
by~\eqref{eq-hardy-mono},
$h^{\omega^{n+1}}(N_0)\leq h^{\omega^\omega}(h(|\?G|))$.  We have just
shown the following upper bound.

\begin{lemma}\label{lem-upb}
  Let $\?G$ be a first-order grammar and $n$, $s$, and $g$ be defined
  as in \cref{th-el}.  Then $\?E_{\?B_{n,s,g}}\leq
  h^{\omega^{\omega}}(h(|\?G|))$ where $h(x)\eqdef H^{\omega^2\cdot
  d}(x)$ for some constant~$d$.
\end{lemma}

\subsection{Fast-Growing Complexity}

It remains to combine \cref{eq-bel} with \cref{lem-upb} in order to
provide an upper bound for the bisimilarity problem.  We will employ
for this the \emph{fast-growing} complexity classes defined
in~\citep{schmitz16}.  This is an ordinal-indexed hierarchy of
complexity classes $(\F\alpha)_{\alpha<\varepsilon_0}$, 
that uses the Hardy functions
$(H^\alpha)_\alpha$ relative to $H(x)\eqdef x+1$ as a standard against
which we can measure high complexities.

\paragraph{Fast-Growing Complexity Classes}
Let us first define
\begin{align}
  \FGH\alpha&\eqdef\bigcup_{\beta<\omega^\alpha}\ComplexityFont{FDTIME}\big(H^\beta(n)\big)
  \intertext{%
    as the class of functions computed by deterministic Turing
    machines in time $O(H^\beta(n))$ for some $\beta<\omega^\alpha$.
    This captures for instance the class of Kalmar elementary
    functions as $\FGH 3$ and the class of primitive-recursive
    functions as $\FGH\omega$~\citep{lob70,wainer72}.  Then we let}
  \F\alpha&\eqdef\bigcup_{p\in\FGH\alpha}\ComplexityFont{DTIME}\big(H^{\omega^\alpha}\!(p(n))\big)
\end{align}
denote the class of decision problems solved by deterministic Turing
machines in time $O\big(H^{\omega^\alpha}\!(p(n))\big)$ for some
function~$p\in\FGH\alpha$.  The intuition behind this quantification
over~$p$ is that, just like e.g.\
$\ComplexityFont{EXPTIME}=\bigcup_{p\in\poly}\ComplexityFont{DTIME}\big(2^{p(n)}\big)$
quantifies
over polynomial functions to provide enough `wiggle room' to account
for polynomial reductions, $\F\alpha$ is closed under $\FGH\alpha$
reductions~\citep[Thms.~4.7 and~4.8]{schmitz16}.
\begin{figure}[tbp]
  \centering\scalebox{.87}{
  \begin{tikzpicture}[every node/.style={font=\small}]
    \shadedraw[color=black!90,top color=black!20,middle
    color=black!5,opacity=20,shading angle=-15](-1,0) arc (180:0:4.8cm);
    \draw[color=black!90,thick,fill=black!10](-.7,0) arc (180:0:3.8cm);
    \draw[color=black!90,fill=black!7,thick](-.65,0) arc (180:0:3.5cm);
    \draw[color=violet!90!black,fill=violet!20,thick](-.6,0) arc (180:0:3.25cm);
    \shadedraw[color=black!90,top color=black!20,middle color=black!5,opacity=20,shading angle=-15](-.5,0) arc (180:0:3cm);
    \draw[color=blue!90,fill=blue!20,thick](-.1,0) arc (180:0:1.7cm);
    \shadedraw[color=black!90,top color=black!20,middle
    color=black!5,opacity=20,shading angle=-15,thin](0,0) arc (180:0:1.5cm);
    \draw (1.5,.5) node{$\ComplexityFont{ELEMENTARY}$};
    \draw (4,1.2) node[color=blue]{$\F3=\!\TOWER$};
    \draw[color=blue!90,thick] (3.15,1) -- (3.05,.9);
    \draw (2.5,1.9) node{$\bigcup_k\!\F{k}{=}\ComplexityFont{PRIMITIVE\text-RECURSIVE}$};
    \draw (5.32,1.5) node[color=violet!90!black]{$\F\omega$};
    \draw (5.73,1.6) node[color=black!70]{$\F{\!\omega^{\!2}}$};
    \draw (6.21,1.7) node[color=black!60]{$\F{\!\omega^3}$};
    \draw (3.95,4) node{$\bigcup_k\!\F{\omega^{k}}=\ComplexityFont{MULTIPLY\text-RECURSIVE}$};
    \draw[loosely
    dotted,very thick,color=black!70](6.7,1.8) --
    (7.2,1.92); \end{tikzpicture}} \caption{Pinpointing
    $\F{\omega}=\ACK$ among the complexity
    classes beyond \ComplexityFont{ELEMENTARY}~\citep{schmitz16}.\label{fig-fg}}
\end{figure}
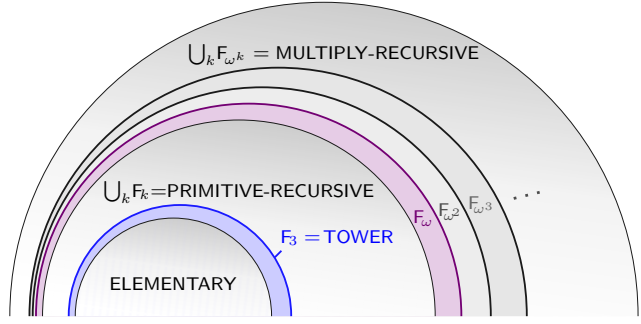

For instance, $\TOWER\eqdef\F 3$ defines the class of
problems that can be solved using computational resources bounded by a
tower of exponentials of elementary height in the size of the input,
$\bigcup_{k\in\+N}\F k$ is the class of primitive-recursive decision
problems, and $\ACK\eqdef\F\omega$ is the class
of problems that can be solved using computational resources bounded
by the Ackermann function applied to some primitive-recursive function
of the input size---here it does not matter for $\alpha>2$ whether we
are considering deterministic, nondeterministic, alternating, time, or
space bounds~\citep[Sec.~4.2.1]{schmitz16}.  See \cref{fig-fg} for a
depiction.

\begin{theorem}\label{th-upb}
  The bisimulation problem for first-order grammars is in
  \ACK, and in $\F{n+4}$ if~$n$ is fixed.
\end{theorem}
\begin{proof}
  This is a consequence of \cref{eq-bel} combined with
  \cref{th-el,lem-upb}; the various overheads on top of the
  bound on~$\?E_{\?B_{n,s,g}}$ are of course negligible for such
  high complexities~\citep[Lem.~4.6]{schmitz16}.  We rely here on
  \citep[Thm.~4.2]{schmitz16} to translate from $h^{\omega^{n+1}}$
  with $h=H^{\omega^2\cdot d}\in\FGH 3$ into a bound in terms
  of~$H^{\omega^{n+4}}$.
\end{proof}

\section{Pushdown Processes}
\label{sec-pda}
The complexity upper bounds obtained in \cref{sec-upb} are stated in
terms of first-order grammars.  In this section, we revisit the known
reduction from pushdown systems to first-order grammars (as given
in~\citep{jancar12,jancar16}), and we also give a direct reduction
from first-order grammars to pushdown systems (instead of giving just
a general reference to~\citep{CourcelleHandbook,caucal95}).  We do
this first to make clear that the reductions are primitive recursive
(in fact, they are polynomial-time reductions), and second to show
that, in the real-time case, \cref{th-upb} provides
primitive-recursive bounds for pushdown systems with a fixed number of
states.

\paragraph{Pushdown Systems}
Let us first recall that a \emph{pushdown system} (\emph{PDS}) is a
tuple $M=\tup{Q,\Sigma,\Gamma,\Delta}$ of finite sets where the
elements of $Q,\Sigma,\Gamma$ are called \emph{control states},
\emph{actions} (or \emph{terminal letters}), and \emph{stack symbols},
respectively; $\Delta$ contains \emph{transition rules} of the form
$pY \step{a}q\gamma$ where $p,q\in Q$, $Y\in\Gamma$,
$a\in \Sigma\uplus\{\varepsilon\}$, and $\gamma\in \Gamma^\ast$.  A
pushdown system is called \emph{real-time} if $a$ is restricted to be
in~$\Sigma$, i.e., if no $\varepsilon$ transition rules appear
in~$\Delta$. 

A PDS $M=\tup{Q,\Sigma,\Gamma,\Delta}$ generates the labelled
transition system
\begin{equation*}
\?L_M\eqdef(Q\times \Gamma^\ast,\Sigma\uplus\{\varepsilon\},
(\step{a})_{a\in\Sigma\cup\{\varepsilon\}})
\end{equation*}
where each rule $pY \step{a}q\gamma$ induces transitions
$pY\gamma' \step{a}q\gamma\gamma'$ for all $\gamma'\in \Gamma^\ast$.
Note that~$\?L_M$ might feature \emph{$\varepsilon$-transitions} (also
called \emph{$\varepsilon$-steps})
$pY\gamma'\step\varepsilon q\gamma\gamma'$ if the PDS is not real-time.

\subsection{From PDS to First-Order Grammars}

We recall 
a construction already presented in the appendix of the
extended version of~\citep{jancar16}.  The idea is that, although
first-order grammars lack the notion of control state, the behaviour
of a pushdown system can nevertheless be captured by a first-order
grammar that uses $m$-ary terms where $m$~is the number of control
states.

\begin{figure}[tbp]
\centering\scalebox{.9}{\hspace*{-1pt}%
  \begin{tikzpicture}[auto,on grid]
    \node[square](A) {$A$};
    \node[square,left=.43 of A]{$p$};
    \node[square,below=.9 of A](C){$C$};
    \node[square,below=.9 of C](B){$B$};
    \path[every node/.style={font=\tiny,inner sep=1pt,color=black!70}]
      (A) edge (C)
      (C) edge (B);
    \node[lsquare,right=2 of A](pA){$pA$};
    \node[lsquare,below left =.9 and .8 of pA](q1C){$q_1C$};  
    \node[lsquare,below      =.9        of pA](q2C){$q_2C$};  
    \node[lsquare,below right=.9 and .8 of pA](q3C){$q_3C$}; 
    \node[lsquare,below left =.9 and .8 of q2C](q1B){$q_1B$};  
    \node[lsquare,below      =.9        of q2C](q2B){$q_2B$};  
    \node[lsquare,below right=.9 and .8 of q2C](q3B){$q_3B$}; 
    \node[lsquare,below left =.9 and .8 of q2B](q1){$q_1$};  
    \node[lsquare,below      =.9        of q2B](q2){$q_2$};  
    \node[lsquare,below right=.9 and .8 of q2B](q3){$q_3$};  
    \path[every node/.style={font=\tiny,inner sep=1pt,color=black!70}]
      (pA) edge (q1C)
      (pA) edge (q2C)
      (pA) edge (q3C)
      (q1C) edge (q1B)
      (q1C) edge (q2B)
      (q1C) edge (q3B)
      (q2C) edge (q1B)
      (q2C) edge (q2B)
      (q2C) edge (q3B)
      (q3C) edge (q1B)
      (q3C) edge (q2B)
      (q3C) edge (q3B)
      (q1B) edge (q1)
      (q1B) edge (q2)
      (q1B) edge (q3)
      (q2B) edge (q1)
      (q2B) edge (q2)
      (q2B) edge (q3)
      (q3B) edge (q1)
      (q3B) edge (q2)
      (q3B) edge (q3);
    \node[square,right=3.5 of pA](a) {$A$};
    \node[square,left=.43 of a]{$p$};
    \node[right=.8 of a]{$\step a$};
    \node[square,right=2 of a](c){$C$};
    \node[square,left=.43 of c]{$q$};
    \node[square,below=.9 of c](ca){$A$};
    \path[every node/.style={font=\tiny,inner sep=1pt,color=black!70}]
      (c) edge (ca);
    \node[lsquare,below left =1.8 and .4 of a](pa){$pA$};
    \node[lsquare,below left =.9 and .8 of pa](x1){$x_1$};  
    \node[lsquare,below      =.9        of pa](x2){$x_2$};  
    \node[lsquare,below right=.9 and .8 of pa](x3){$x_3$}; 
    \node[right=1.2 of pa]{$\step a$};
    \node[lsquare,right=2.6 of pa](qc){$qC$};
    \node[lsquare,below left =.9 and .8 of qc](q1a){$q_1A$};  
    \node[lsquare,below      =.9        of qc](q2a){$q_2A$};  
    \node[lsquare,below right=.9 and .8 of qc](q3a){$q_3A$}; 
    \node[lsquare,below left =.9 and .8 of q2a](x1a){$x_1$};  
    \node[lsquare,below      =.9        of q2a](x2a){$x_2$};  
    \node[lsquare,below right=.9 and .8 of q2a](x3a){$x_3$};   
    \path[every node/.style={font=\tiny,inner sep=1pt,color=black!70}]
      (pa) edge (x1)
      (pa) edge (x2)
      (pa) edge (x3)
      (qc) edge (q1a)
      (qc) edge (q2a)
      (qc) edge (q3a)
      (q1a) edge (x1a)
      (q1a) edge (x2a)
      (q1a) edge (x3a)
      (q2a) edge (x1a)
      (q2a) edge (x2a)
      (q2a) edge (x3a)
      (q3a) edge (x1a)
      (q3a) edge (x2a)
      (q3a) edge (x3a);    
    \node[right=.55 of C]{\Large$\rightsquigarrow$};
    \node[below left=.37 and 1.2 of ca,rotate=-90]{\Large$\rightsquigarrow$};
  \end{tikzpicture}}
\caption{The PDS configuration $pACB$ encoded as a term (left), and
  the translation of the PDS rule $pA\step a qCA$ into a
  first-order rule (right).}\label{fig:pdatofo}
\end{figure}
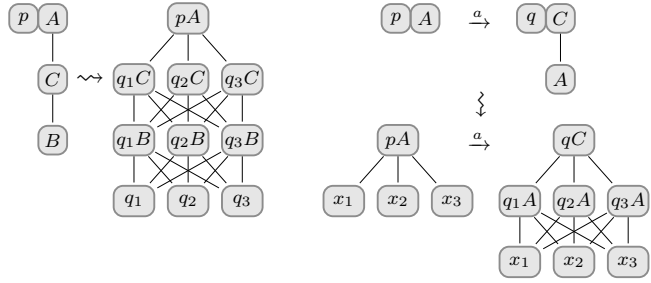

\Cref{fig:pdatofo} (left) presents a configuration of a PDS---i.e., a
state in $\?L_M$---as a term; here we assume that $Q=\{q_1,q_2,q_3\}$.
The string $pACB$, depicted on the left in a convenient vertical form,
is translated into a term presented by an acyclic graph in the
figure.  On the right in \cref{fig:pdatofo} we can see the translation
of the PDS transition rule $pA\step{a}qCA$ into a rule of a first-order
grammar.

\subsubsection{Real-Time Case}

Let us first assume that $M$ is a real-time PDS, i.e., that each PDS
transition rule $pY\step a q\gamma$ is such that $a$ is in~$\Sigma$.
We are interested in the following decision problem.
\begin{problem}[Strong Bisimulation]
\hfill\\[-1.5em]\begin{description}[\IEEEsetlabelwidth{question}]
\item[input] A real-time pushdown system
  $M=\tup{Q,\Sigma,\Gamma,\Delta}$ and two configurations $pY,qZ$
  in~$Q\times\Gamma$.
\item[question] Is $pY\sim qZ$ in the labelled transition
  system~$\?L_M$?
\end{description}\end{problem}

Formally, for a real-time PDS $M=(Q,\Sigma,\Gamma,\Delta)$, where
$Q=\{q_1,q_2,\dots,q_m\}$, we can define the first-order grammar
\begin{equation*}
 \?G_M\eqdef(\N,\Sigma,\ru)
\end{equation*}
where $\N\eqdef Q\cup (Q\times \Gamma)$,
with $\ar(q)\eqdef 0$ and $\ar((q,X))\eqdef m$ for all $q$ in~$Q$ and
$X$ in~$\Gamma$; the set $\ru$ is defined below.  We write $[q]$ and
$[qY]$ for nonterminals $q$ and $(q,Y)$, respectively, and we map each
configuration $p\gamma$ to a (finite) term $\?T(p\gamma)$ in $\terms$ defined
by structural induction:
\begin{align}
  \?T(p\varepsilon)&\eqdef[p]\;,\label{eq-Teps}\\
  \?T(pY\gamma)&\eqdef
                 [pY](\?T(q_1\gamma),\?T(q_2\gamma),\dots,\?T(q_m\gamma))\;.\label{eq-TY}
\intertext{For a smooth translation of rules, we introduce a special `stack
variable' $x$, and we set }\?T(q_ix)&\eqdef x_i\label{eq-Tx}
\end{align} for all $i\in\{1,\dots,m\}$.

A PDS transition rule $pY\step{a}q\gamma$
in~$\Delta$ with $a$ in~$\Sigma$ is then translated into the
first-order grammar rule
\begin{align}\label{eq-ru}
\?T(pYx)&\step{a}\?T(q\gamma x)
\intertext{%
in~$\ru$.  Hence $pY\step{a}q_i$ is
translated into}
\notag
[pY](x_1,\dots,x_m)&\step{a}x_i
\intertext{%
and
$pY\step{a}qZ\gamma$ is translated into}
\notag
[pY](x_1,\dots,x_m)&\step{a} [qZ](\?T(q_1\gamma x),\dots,\?T(q_m\gamma
x))\;.
\end{align}

It should be obvious that the labelled transition system~$\?L_M$ is
isomorphic with the restriction of the labelled transition system
$\?L_{\?G_M}$
to the states $\?T(p\gamma)$ where $p\gamma$
are configurations of~$M$; moreover, the set
$\{\?T(p\gamma)\mid p\in Q, \gamma\in\Gamma^\ast\}$ is closed w.r.t.\
reachability in 
$\?L_{\?G_M}$: %
if $\?T(p\gamma)\step{a}F$
in $\?L_{\?G_M}$, then $F=\?T(q\gamma')$ where $p\gamma\step{a}q\gamma'$
in $\?L_M$.

\begin{corollary}\label{cor-pds}
  The strong bisimulation problem for real-time pushdown systems is in
  \ACK, and in $\F{|Q|+4}$ if the number $|Q|$ of states is fixed.
\end{corollary}
\begin{proof}
  What we have sketched above is a polynomial-time (in fact,
  \ComplexityFont{logspace}) reduction from the strong bisimulation
  problem in (real-time) pushdown systems to the bisimulation problem
  in first-order grammars, for which we can apply \cref{th-upb}.
  Observe that, in this translation and according to the discussion
  after~\eqref{eq-n}, we may bound~$n$ by the number~$|Q|$ of states
  of the given pushdown system, which justifies the
  primitive-recursive $\F{|Q|+4}$ upper bound when the number of
  states is fixed.  (\Cref{fig:pdatofo} makes clear that all branches
  in $\?T(p\gamma)$ have the same lengths, and there are
  precisely~$|Q|$ depth-$d$ subterms of $\?T(p\gamma)$, for each
  $d\leq\height{\?T(p\gamma)}$.)
\end{proof}

\subsubsection{General Case}
In the case of labelled transition systems
$\?L=\tup{\?S,\Sigma,({\step a})_{a\in\Sigma\uplus\{\varepsilon\}}}$
with a \emph{silent action}~$\varepsilon$,
by $s\dstep w t$,
for $w\in\Sigma^\ast$, we denote 
that there are $s_0,s_1,\dots,s_\ell\in \?S$ and
$a_1,\dots,a_\ell\in\Sigma\uplus\{\varepsilon\}$
such that 
$s_0=s$, $s_\ell=t$, $s_{i-1}\step{a_i}s_i$ for all
$i\in\{1,\dots,\ell\}$, and $w=a_1\cdots a_\ell$.
Thus $s\dstep\varepsilon t$ denotes
an arbitrary sequence of silent steps, and 
$s\dstep{a} t$ for $a\in\Sigma$ denotes that there are $s',t'$ such
that $s\dstep{\varepsilon}s'\step{a}t'\dstep{\varepsilon}t$.

A relation $R\subseteq\?S\times\?S$ is a \emph{weak bisimulation}
if the following two conditions hold:
\begin{description}[\IEEEsetlabelwidth{(zag)}]
\item[(zig)] if $s\mathbin R t$ and $s\step a s'$ for some
  $a\in\Sigma\uplus\{\varepsilon\}$, then there exists $t'$ such that
  $t\dstep a t'$ and $s'\mathbin R t'$;
\item[(zag)] if $s\mathbin R t$ and $t\step a t'$ for some
  $a\in\Sigma\uplus\{\varepsilon\}$, then there exists $s'$ such that
  $s\dstep a s'$ and $s'\mathbin R t'$.
\end{description}
By $\approx$ we denote \emph{weak bisimilarity}, i.e., the largest
weak bisimulation (the union of all weak bisimulations), which is an
equivalence relation.

We are now interested in the following problem.
\begin{problem}[Weak Bisimulation]
\hfill\\[-1.5em]\begin{description}[\IEEEsetlabelwidth{question}]
\item[input] A pushdown system
  $M=\tup{Q,\Sigma,\Gamma,\Delta}$ and two configurations $pY,qZ$
  in~$Q\times\Gamma^\ast$.
\item[question] Is $pY\approx qZ$ in the labelled transition
  system~$\?L_M$?
\end{description}\end{problem}

Unfortunately, in general the weak bisimulation problem for PDS is
undecidable, already for one-counter systems~\cite{mayr03}; we can
also refer, e.g., to~\cite{jancar08} for further discussion.  As
already mentioned in the introduction, we thus consider PDS with
(very) \emph{restricted silent actions}: each rule
$pY\step{\varepsilon}q\gamma$ in $\Delta$ is \emph{deterministic}
(i.e., alternative-free), which means that there is no other rule with
the left-hand side $pY$.
From now on, by \emph{restricted PDS} we mean PDS whose
$\varepsilon$-rules are deterministic.

We aim to show that the weak bisimulation problem for restricted PDS
reduces to the (strong) bisimulation problem for first-order grammars
(where silent actions are not allowed by our definition). For this it
is convenient to make a standard transformation~\citep[see,
e.g.,][Sec.~5.6]{harrison78} of our restricted PDS that removes
non-popping $\varepsilon$-rules; an $\varepsilon$-rule
$pY\step{\varepsilon}q\gamma$ is called \emph{popping} if
$\gamma=\varepsilon$.  This is captured by the next proposition.
(When comparing two states from different LTSs, we implicitly refer to
the disjoint union of these LTSs.)

\begin{proposition}\label{prop:restrtopop}
There is a polynomial-time transformation
of a restricted  PDS $M=\tup{Q,\Sigma,\Gamma,\Delta}$ to 
$M'=\tup{Q,\Sigma,\Gamma,\Delta'}$ in which each $\varepsilon$-rule is
deterministic and popping, and $pY$ in $\?L_M$ is weakly bisimilar 
with $pY$ in $\?L_{M'}$.
\end{proposition}	
\begin{proof}
Given a restricted PDS $M=\tup{Q,\Sigma,\Gamma,\Delta}$, we 
proceed as follows.
First we find all $pY$ such that
\begin{align}\label{pop1}
  pY&\step{\varepsilon}\cdots\dstep{\varepsilon}pY\gamma
      \intertext{for some $\gamma\in\Gamma^\ast$, and remove the respective 
      rules $pY\step{\varepsilon}\cdots$.
      Then for each $pY$ such that}\label{pop2}
  pY&\step{\varepsilon}\cdots\step{\varepsilon}\cdots\dstep{\varepsilon}q\;,
      \intertext{we add the popping rule
      $pY\step{\varepsilon}q$, and for each 
      $pY$ where}
 pY&\step{\varepsilon}\cdots\dstep{\varepsilon}qB\gamma\label{pop3}
\end{align}
and each rule $qB\step{a}q'\gamma'$ with $a\in\Sigma$ we add the rule
$pY\step{a}q'\gamma'\gamma$.  Finally we remove all the non-popping
$\varepsilon$-rules.  Thus $M'=\tup{Q,\Sigma,\Gamma,\Delta'}$ arises.
Identifying the configurations that satisfy conditions
(\ref{pop1}--\ref{pop3}) can be performed in polynomial time through a
saturation algorithm.
The claim on the relation of $\?L_M$ and $\?L_{M'}$ is
straightforward.
\end{proof}
A \emph{stable configuration} is either a configuration
$p\varepsilon$, or a configuration $pY\gamma$ where there is no 
$\varepsilon$-rule of the form
$pY\step{\varepsilon}q\gamma'$.
In a
restricted PDS with only popping $\varepsilon$-rules, any unstable
configuration $p\gamma$ only allows to perform a finite sequence of
silent popping steps until it reaches a stable configuration.
It is
natural to restrict our attention to the transitions
$p\gamma\step{a}q\gamma'$ with $a\in\Sigma$ between stable
configurations; such transitions might encompass sequences of
popping $\varepsilon$-steps.

When defining the grammar $\?G_M$, we can avoid the explicit use of
deterministic popping silent steps, by `preprocessing'
them: we apply the inductive definition of the translation
operator~$\?T$ from (\ref{eq-Teps}--\ref{eq-Tx}) to stable
configurations, while if $pY$ is unstable, then there is exactly one
applicable rule,
$pY\step{\varepsilon}q$, and in this case we let
\begin{equation}\label{eq-Tuns}
 \?T(pY\gamma)\eqdef\?T(q\gamma)\;.
\end{equation}
\begin{figure}[tbp]
\centering
\centering\scalebox{.9}{\hspace*{-1pt}%
  \begin{tikzpicture}[auto,on grid]
    \node[square](a) {$A$};
    \node[square,left=.43 of a]{$p$};
    \node[right=.6 of a]{$\step a$};
    \node[square,right=1.6 of a](c){$C$};
    \node[square,left=.43 of c]{$q$};
    \node[square,below=.9 of c](ca){$A$};
    \path[every node/.style={font=\tiny,inner sep=1pt,color=black!70}]
      (c) edge (ca);
    \node[square,below=1.8 of a](A) {$A$};
    \node[square,left=.43 of A]{$q_2$};
    \node[right=.6 of A]{$\step\varepsilon$};
    \node[square,right=1.17 of A](q3){$q_3$};
    \node[lsquare,right=4.2 of a](pa){$pA$};
    \node[lsquare,below left =.9 and .8 of pa](x1){$x_1$}; 
    \node[lsquare,below      =.9        of pa](x2){$x_2$};  
    \node[lsquare,below right=.9 and .8 of pa](x3){$x_3$}; 
    \node[right=1.2 of pa]{$\step a$};
    \node[lsquare,right=2.6 of pa](qc){$qC$};
    \node[lsquare,below left =.9 and .8 of qc](q1a){$q_1A$};  
    \node[lsquare,below right=.9 and .8 of qc](q3a){$q_3A$}; 
    \node[lsquare,below left =1.8 and .8 of qc](x1a){$x_1$};  
    \node[lsquare,below      =1.8        of qc](x2a){$x_2$};  
    \node[lsquare,below right=1.8 and .8 of qc](x3a){$x_3$};   
    \path[every node/.style={font=\tiny,inner sep=1pt,color=black!70}]
      (pa) edge (x1)
      (pa) edge (x2)
      (pa) edge (x3)
      (qc) edge (q1a)
      (qc) edge[bend right=15] (x3a)
      (qc) edge (q3a)
      (q1a) edge (x1a)
      (q1a) edge (x2a)
      (q1a) edge (x3a)
      (q3a) edge (x1a)
      (q3a) edge (x2a)
      (q3a) edge (x3a);
    \node[right=.9 of ca]{\Large$\rightsquigarrow$};
  \end{tikzpicture}}
	\caption{Deterministic popping silent steps are
	`preprocessed.'}\label{fig:swalloweps}
\end{figure}
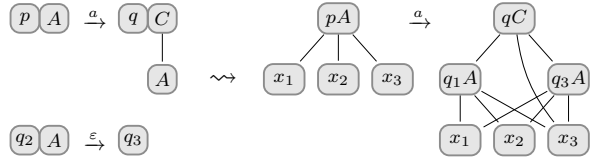%
\Cref{fig:swalloweps} (right) shows the grammar-rule
\begin{equation*}
 \?T(pAx)\step{a}\?T(qCAx)
\end{equation*}
(arising from the PDS-rule
$pA\step{a}qCA$), when $Q=\{q_1,q_2,q_3\}$ and 
there is a PDS-rule $q_2A\step{\varepsilon}q_3$, while $q_1A$, $q_3A$
are stable.

\begin{corollary}\label{cor-e-pds}
	The weak bisimulation problem for restricted pushdown systems (i.e.,
	where $\varepsilon$-rules are deterministic) is in~$\ACK$.
\end{corollary}
\begin{proof}
	By Proposition~\ref{prop:restrtopop} it suffices to consider 
	a PDS $M=(Q,\Sigma,\Gamma,\Delta)$ where each
	$\varepsilon$-rule is deterministic and popping.
Since it is clear that $pY\approx qZ$ in $\?L_M$ iff 
$\?T(pY)\sim \?T(qZ)$ in $\?L_{\?G_M}$, the claim follows from
 \cref{th-upb}.
\end{proof}

Note that, due to our preprocessing, the terms~$\?T(p\gamma)$ may
have branches of varying lengths, which is why~$n$ as defined
in~\eqref{eq-n} might not be bounded by the number of states as in
\cref{cor-pds}.

\subsection{From First-Order Grammars to PDS}

We have shown the $\ACK$-membership for bisimilarity of first-order
grammars (\cref{th-upb}), and thus also for weak bisimilarity of
pushdown processes with deterministic $\varepsilon$-steps
(\cref{cor-e-pds}).  By adding the lower bound from~\cite{jancarhard},
we get the $\ACK$-completeness for both problems.

In fact, the $\ACK$-hardness in~\cite{jancarhard} was shown in the
framework of first-order grammars. The case of pushdown processes was
handled by a general reference to the equivalences that are known,
e.g., from~\cite{CourcelleHandbook} and the works referred there;
another relevant reference for such equivalences is~\cite{caucal95}.
Nevertheless, in our context it seems more appropriate to show a
direct transformation from first-order grammars to pushdown processes
(with deterministic $\varepsilon$-steps), which can be argued to be
primitive-recursive;
in fact, it is 
a \ComplexityFont{logspace} reduction.

\medskip 
Let $\?G=\tup{\N,\Sigma,\ru}$ be a first-order grammar.  For
a term $F\in\terms$ such that $F\not\in\X$ (hence the root of $F$ is a
nonterminal $A$) we define its \emph{root-substitution} to be the
substitution $\sigma$ where $F=A(x_1,\dots,x_{\ar(A)})\sigma$ and
$x\sigma=x$ for all $x\not\in\{x_1,\dots,x_{\ar(A)}\}$.  A
substitution $\sigma$ is an \emph{rhs-substitution} for~$\?G$ if it is
the root-substitution of a subterm $F$ of the right-hand side $E$ of a
rule $A(x_1,\dots,x_{\ar(A)})\step{a}E$ in $\ru$ (where $F\not\in\X$);
we let $\textsc{RSubs}_\?G$ denote the set of
rhs-substitutions for~$\?G$.

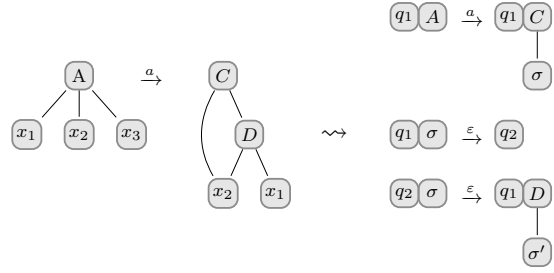
\begin{figure}[!t]
  \centering\scalebox{.87}{
  \begin{tikzpicture}[auto,on grid]
    \node[square](A){A};
    \node[square,below left =.9 and .8 of A](x1){$x_1$}; 
    \node[square,below      =.9        of A](x2){$x_2$};  
    \node[square,below right=.9 and .8 of A](x3){$x_3$}; 
    \node[right=1.1 of A]{$\step a$};
    \node[square,right=2.2 of A](C){$C$};
    \node[square,below right=.9 and .4 of C](D){$D$}; 
    \node[square,below left =.9 and .4 of D](C1){$x_2$}; 
    \node[square,below right=.9 and .4 of D](D2){$x_1$};
    \path[every node/.style={font=\tiny,inner sep=1pt,color=black!70}]
      (A) edge (x1)
      (A) edge (x2)
      (A) edge (x3)
      (C) edge[bend right] (C1)
      (C) edge (D)
      (D) edge (C1)
      (D) edge (D2);
    \node[right=1.3 of D]{\Large$\rightsquigarrow$};
    \node[square,above right=.9 and 3.2 of C](a){$A$};
    \node[square,left=.43 of a]{$q_1$};
    \node[right=.6 of a]{$\step a$};
    \node[square,right=1.6 of a](c){$C$};
    \node[square,left=.43 of c]{$q_1$};
    \node[square,below=.9 of c](sub){$\sigma$};
    \node[square, below=1.8 of a](s){$\sigma$};
    \node[square,left=.43 of s]{$q_1$};
    \node[right=.6 of s]{$\step\varepsilon$};
    \node[square,right=1.17 of s](q2){$q_2$};
    \node[square, below=.9 of s](s2){$\sigma$};
    \node[square,left=.43 of s2]{$q_2$};
    \node[right=.6 of s2]{$\step\varepsilon$};
    \node[square,right=1.6 of s2](d){$D$};
    \node[square,left=.43 of d]{$q_1$};
    \node[square,below=.9 of d](sub2){$\sigma'$};
    \path[every node/.style={font=\tiny,inner sep=1pt,color=black!70}]
      (c) edge (sub)
      (d) edge (sub2);
  \end{tikzpicture}}
  \caption{\label{fig-fog2pda}The transformation from first-order
    grammars to pushdown processes with deterministic
    $\varepsilon$-steps.  In this example,
$x_1\sigma=x_2, x_2\sigma=D(x_2,x_1)$, 
and $x_1\sigma'=x_2, x_2\sigma'=x_1$.}
\end{figure}

We define the PDS
  $M_{\?G}\eqdef(Q,\Sigma,\Gamma,\Delta)$
where
\ifams%
\begin{align*}
  Q&\eqdef\{q_1,\dots,q_m\}
\intertext{for $m$ as defined in~\eqref{eq-m}---or
     $Q\eqdef\{q_1\}$ if $m=0$---,} 
 \Gamma&\eqdef \N\uplus\textsc{RSubs}_\?G\;,\\
  \Delta&\eqdef\{q_1A\step a q_i\mid
          (A(x_1,\dots,x_{\ar(A)})\step{a}x_i)\in\ru\}\\
        &\,\cup\,\{q_1A\step{a}q_1B\sigma\mid \sigma\in\textsc{RSubs}_\?G\wedge(A(x_1,\dots,x_{\ar(A)})\step{a}B(x_1,\dots,x_{\ar(B)})\sigma)\in\ru\}\\
        &\,\cup\,\{q_i\sigma\step\varepsilon
          q_j\mid 1\leq i\leq
          m\wedge\sigma\in\textsc{RSubs}_\?G\wedge\sigma(x_i)=x_j\}\\
        &\,\cup\,\{q_i\sigma\step{\varepsilon}q_1C\sigma'\mid 1\leq i\leq
          m\wedge\sigma,\sigma'\in\textsc{RSubs}_\?G\wedge\sigma(x_i)=C(x_1,\dots,x_{\ar(C)})\sigma'\}\;.
\end{align*}\else%
\begin{align*}
  Q&\eqdef\{q_1,\dots,q_m\}
\intertext{for $m$ as defined in~\eqref{eq-m}---or
     $Q\eqdef\{q_1\}$ if $m=0$---,} 
 \Gamma&\eqdef \N\uplus\textsc{RSubs}_\?G\;,\\
  \Delta&\eqdef\{q_1A\step a q_i\mid
          (A(x_1,\dots,x_{\ar(A)})\step{a}x_i)\in\ru\}\\
        &\,\cup\,\{q_1A\step{a}q_1B\sigma\mid \sigma\in\textsc{RSubs}_\?G\:\wedge\\
          &\qquad\qquad(A(x_1,\dots,x_{\ar(A)})\step{a}B(x_1,\dots,x_{\ar(B)})\sigma)\in\ru\}\\
        &\,\cup\,\{q_i\sigma\step\varepsilon
          q_j\mid 1\leq i\leq
          m\wedge\sigma\in\textsc{RSubs}_\?G\wedge\sigma(x_i)=x_j\}\\
        &\,\cup\,\{q_i\sigma\step{\varepsilon}q_1C\sigma'\mid 1\leq i\leq
          m\wedge\sigma,\sigma'\in\textsc{RSubs}_\?G\:\wedge\\
        &\qquad\qquad\sigma(x_i)=C(x_1,\dots,x_{\ar(C)})\sigma'\}\;.
\end{align*}\fi%
See \cref{fig-fog2pda} for an example.  Note that the $\varepsilon$-rules
are indeed deterministic; moreover, any non-popping
$\varepsilon$-step, hence of the form
$q_i\sigma\gamma\step{\varepsilon}q_1C\sigma'\gamma$, cannot be
followed by another $\varepsilon$-step.

It should be obvious that a state $A(x_1,\dots,x_{\ar(A)})$ in $\lts$
is weakly bisimilar with the state $q_1A$ in $\?L_{M_{\?G}}$.
In particular we note that $q_1A\dstep{w}q_i\gamma$ in $\?L_{M_{\?G}}$ 
(where also $\varepsilon$-steps might be comprised) 
entails 
that $\gamma=\sigma_0\sigma_1\dots\sigma_\ell$
(in which case $q_i\gamma$ represents the term $x_i\sigma_0\sigma_1,\dots\sigma_\ell$), or 
 $\gamma=B\sigma_1\dots\sigma_\ell$ when $i=1$
 (in which case $q_1\gamma$ represents the term
 $B(x_1,\dots,x_{\ar(B)})\sigma_1,\dots\sigma_\ell$).

We could add a technical discussion about how to represent all the
terms from~$\terms$ (including the infinite regular terms) in an
enhanced version of~$\?L_{M_{\?G}}$, but this is not necessary since
the lower bound construction in~\cite{jancarhard} uses only the
states of~$\lts$ that are reachable from `initial' terms of the
form~$A(x_1,\dots,x_{\ar(A)})$ (more precisely, of the form
$A(\bot,\dots,\bot)$ for a nullary nonterminal~$\bot$).

\begin{corollary}\label{cor:pdshard}
The weak bisimulation problem for 
 pushdown systems whose $\varepsilon$-rules are deterministic and
 popping is~$\ACK$-hard.
\end{corollary}
\begin{proof}
In~\citep{jancarhard}, the
\ACK-hardness of the control-state reachability problem for reset
counter machines is recalled~\citep{schnoebelen10}, and its
polynomial-time
(in fact, \ComplexityFont{logspace}) reduction to the bisimulation problem for
first-order grammars is shown.
The reduction guarantees that a given control state is reachable from the initial
configuration of a given reset counter machine $R$ iff 
$A(\bot,\dots,\bot)\not\sim B(\bot,\dots,\bot)$
in $\?L_{\?G_R}$ for the constructed grammar $\?G_R$.
As shown above, the question whether $A(\bot,\dots,\bot)\sim B(\bot,\dots,\bot)$ 
in $\?L_{\?G_R}$
can be
further reduced to an instance of the weak bisimulation problem for
the pushdown system $M_{\?G_{R}}$. 
\end{proof}

\section{Concluding Remarks}
\label{sec-concl}
\Cref{th-upb,cor-e-pds} provide the first known worst-case upper
bounds, in \ACK, for the strong bisimulation equivalence of
first-order grammars and the weak bisimulation equivalence of pushdown
processes restricted to 
deterministic silent steps.  By the lower
bound shown in~\citep{jancarhard} and \cref{cor:pdshard}, this is moreover optimal.  An
obvious remaining problem is to close the complexity gap in the case of
strong bisimulation for real-time pushdown processes, which is only
known to be \TOWER-hard~\citep{benedikt13}, and for which we do not
expect \cref{cor-pds} to provide tight upper bounds.

\appendix
\ifams\section{Grammatical Constants}\fi
\label{app-cst}

\renewcommand{\eqby}[1]{=}
\begin{table*}[!t]
  \caption{\ifams Grammatical constants defined
    in~\citep{jancar18}.\else Grammatical Constants Defined in~\citep{jancar18}\fi}
  \label{tab-cst}
  \centering\ifams\footnotesize\setlength{\tabcolsep}{2pt}\fi
  \begin{tabular}{r@{$\;\eqdef\;$}lccl}
    \toprule
    \multicolumn{2}{l}{Constant} & Ref.\ in~\citep{jancar18} &
    Ref.\ \ifams here\else in this paper\fi & Growth in~$|\?G|$\\
    \midrule
    $m$&$\max_{A\in\N}\ar(A)$  & (7) & \eqref{eq-m} & linear\\
    $\hinc$&$\max_{E\in\rhs}\height{E}-1$ & (4) & \eqref{eq-hinc} & linear\\
    $\sinc$&$\max_{E\in\rhs}\ntsize{E}$ & (5) & \eqref{eq-sinc} & linear\\
    $d_0$&$ 1+\max_{A\in\N,1\leq i\leq\ar(A)}|w_{[A,i]}|$ & (6) &
    \eqref{eq-d0} & exponential\\
    $d_1$&$ 2|\N|(\max\{d_0,|\ru|^{d_0}\})^{m+2}$ & (13) & & doubly
    exponential\\
    $d_2$&$ d_0+(1+d_0\hinc)(d_0-1)$ & (19) & & exponential\\
    $d_3$&$(\max\{d_0,|\ru|^{d_0}\})^2$ & (21) && doubly
    exponential\\
    $n$&$m^{d_0}$ & (24) &
    \eqref{eq-n} & doubly exponential\\
    $s$&$ m^{d_0+1}+(m+2)d_0\sinc+(d_2+d_0-1)\sinc$ & (25) & &
    doubly exponential\\
    $g$&$(d_2+d_0-1)\sinc$ & (26) & & exponential\\
    $d_4$&$d_1(1+\sum_{E\in\rhs}\ntsize{E})^{d_2+d_0-1}$ & (23) & &
    doubly exponential\\
    $d_5$&$(d_2+d_0-1)(1+(d_0-1)\hinc)$ & (31) && doubly exponential\\
    $c$&$\max\{d_3,2d_4d_5\}$ & (38) && doubly exponential\\
    \bottomrule
  \end{tabular}
\end{table*}

The proof of \cref{th-el} in~\citep[Thm.~7]{jancar18} relies on the
definition of several grammatical constants, which depend solely on the given
first-order grammar~$\?G=\tup{\N,\Sigma,\ru}$.  
In \cref{tab-cst} we summarise their definitions as a reference for the reader.

\section*{Acknowledgements}
\addcontentsline{toc}{section}{Acknowledgements} 
{P. Jan\v{c}ar acknowledges the support of the 
 Grant Agency of Czech Rep., GA\v{C}R 18-11193S;
 part of this research 
 was conducted while he held an invited professorship
 at ENS Paris-Saclay.
S. Schmitz is partially funded by
  \href{http://bravas.labri.fr/}{ANR-17-CE40-0028~\textsc{Bra\!VAS}}.}

\newcommand{\noopsort}[2]{#2}\newcommand{\complexity}[1]{{\smaller\textsf{#1}}}

\end{document}